\documentclass{article}

\usepackage{arxiv}

\usepackage[utf8]{inputenc} 
\usepackage[T1]{fontenc}    
\usepackage{hyperref}       
\usepackage{url}            
\usepackage{booktabs}       
\usepackage{amsfonts}       
\usepackage{nicefrac}       
\usepackage{microtype}      
\usepackage{lipsum}
\usepackage{fancyhdr}       
\usepackage{graphicx}       
\graphicspath{{media/}}     
\usepackage{amsmath}
\usepackage{subcaption}
\usepackage{amssymb}
\usepackage{amsthm}
\usepackage{mathtools}
\usepackage{pdfpages}
\usepackage{float}
\newtheorem{theorem}{Theorem}

\newtheorem{lemma}{Lemma}
\newtheorem{definition}{Definition}
\newtheorem{remark}{Remark}

\newtheorem{assumption}{Assumption}
\newtheorem{problem}{Problem}

\usepackage{algorithm}
\usepackage{algpseudocode}
\newcommand\ignore[1]{{}}
\algrenewcommand\algorithmicrequire{\textbf{Input:}}
\algrenewcommand\algorithmicensure{\textbf{Output:}}
\DeclareUnicodeCharacter{221E}{-}
\usepackage{mathrsfs}

\title{Model-Free Incremental Adaptive Dynamic Programming Based Approximate Robust Optimal Regulation 
}

\author{
  Cong Li \\
     \texttt{cong.li@tum.de} \\
  \And
 Yongchao Wang \\
 \texttt{yongchao.wang@tum.de} \\
    \And
     Fangzhou Liu \\
 \texttt{fangzhou.liu@tum.de} \\
     \And
     Qingchen Liu \\
 \texttt{qingchen.liu@tum.de} \\
   \And
 Martin Buss \\
\texttt{mb@tum.de} \\
}

\begin{document}
\maketitle

\begin{abstract}
This paper presents a new formulation for model-free robust optimal regulation of continuous-time nonlinear systems. The proposed reinforcement learning based approach, referred to as incremental adaptive dynamic programming (IADP), exploits measured data to allow the design of the approximate optimal incremental control strategy, which stabilizes the controlled system incrementally under model uncertainties, environmental disturbances, and input saturation. By leveraging the time delay estimation (TDE) technique, we first exploit sensory data to reduce the requirement of a complete dynamics, where measured data are adopted to construct an incremental dynamics that reflects the system evolution in an incremental form. Then, the resulting incremental dynamics serves to design the approximate optimal incremental control strategy based on adaptive dynamic programming, which is implemented as a simplified single critic structure to get the approximate solution to the value function of the Hamilton-Jacobi-Bellman equation. Furthermore, for the critic artificial neural network, experience data are used to design an off-policy weight update law with guaranteed weight convergence. Rather importantly, to address the unintentionally introduced TDE error, we incorporate a TDE error bound related term into the cost function, whereby the TDE error is attenuated during the optimization process.
The system stability proof and the weight convergence proof are provided. Numerical simulations are conducted to validate the effectiveness and superiority of our proposed IADP, especially regarding the reduced control energy expenditure and the enhanced robustness.
\end{abstract}


\section{Introduction}
Reinforcement learning (RL) provides a mathematical formalism for learning-based control strategies and has shown superior performance in multiple scenarios, e.g., robots \cite{kober2013reinforcement}, unmanned aerial vehicles \cite{koch2019reinforcement}, and autonomous driving \cite{kuutti2020survey}, etc.
However, although the distinguishable model-free feature of RL encourages researchers to adopt RL-based control strategies to overcome the difficulty of applying traditional model-based control methods to the unknown (or hardly modeled) plants, the rigorous system stability analysis is not provided in most of the related works, see \cite{recht2019tour, bucsoniu2018reinforcement} and the references therein. A system without stability guarantee may be potentially dangerous \cite{khalil2002nonlinear}.
Recently, synchronous adaptive dynamic programming (ADP) \cite{vamvoudakis2010online,kamalapurkar2015approximate,vamvoudakis2014online}, where actor and critic artificial neural networks (ANNs) update simultaneously in real-time, emerges as a promising control-theoretic RL subfield featured for available system stability proofs. 
However, its provided stability proof trades off the attractive model-free feature of RL since the dynamics is required to present the rigorous system stability analysis. 
Even though the required explicit knowledge of dynamics could be avoided by using
add-on techniques such as NNs \cite{bhasin2013novel,zhao2019event,zhang2011data}, fuzzy models \cite{su2019adaptive}, Gaussian process (GP) \cite{boedecker2014approximate}, or observers \cite{sun2018disturbance}, however, the accompanying identification processes further increase computational complexity and parameter tuning efforts.
This motivates us to develop a novel computationally simple RL-based control strategy, which enjoys both a model-free feature and a provable system stability guarantee, to accomplish the robust optimal stabilization of continuous-time nonlinear systems.
Many attempts have been conducted to enhance the robustness of synchronous ADP.
To approximately solve the robust optimization problem of a completely known dynamics perturbed by an unknown but (assumed) bounded additive disturbance, existing synchronous ADP related approaches are divided into two categories: the zero-sum game method based on the H-infinity theory \cite{vamvoudakis2012online}, and the transformed optimal control method with a well-designed utility function \cite{liu2015reinforcement}.
However, both methods require accurate model information to construct the corresponding Hamilton-Jacobi-Issac (HJI) or Hamilton-Jacobi-Bellman (HJB) equations. Moreover, the worst-case disturbance related terms incorporated into the cost functions usually inevitably result in conservative control policies that lead to reduced performance, especially when controlled plants suffer from disturbances with occasional high peaks. In addition, the transformed optimal control method further demands the structure knowledge of the disturbance, i.e., the disturbance bound.
To obviate the requirement of an accurate drift dynamics, by using the defined integral reinforcement, integral RL is developed to allow the design of a partially model-free approximate optimal control strategy  \cite{vrabie2009neural}. Whereas, the complete knowledge of an input dynamics is still demanded.
A further step to get rid of model information is to exploit NN based approximation schemes such as ANNs \cite{bhasin2013novel}, radial basis function neural networks (RBFNNs) \cite{zhao2019event}, and recurrent neural networks (RNNs) \cite{zhang2011data}, where the dynamics is approximated as a linear weighting of handpicked basis sets.
Albeit the model-free control is achieved based on the universal approximation ability of NNs, it is not trivial to get a high-quality approximated model based on an additionally introduced weight update law. The control strategy designed based on inaccurate approximated models might lead to performance degradation or even instability.
Moreover, the effectiveness of these plug-in methods \cite{bhasin2013novel, zhao2019event,zhang2011data} highly relies on prior experience. For example, the choice of radial basis functions needs trials and errors until the approximation error is under a given threshold \cite{zhao2019event}.
The aforementioned deficiencies such as possible inaccurate approximated models and high reliances on prior experience also exist in the fuzzy model based work \cite{su2019adaptive} to avoid using model information.
In addition, GP \cite{boedecker2014approximate} or observer \cite{sun2018disturbance} based methods are also widely used to deal with model uncertainties and/or environmental disturbances. Although efficient, these add-on methods \cite{boedecker2014approximate,sun2018disturbance} suffer from high computation complexity and parameter tuning efforts due to additionally incorporated identification processes. 
The counterpart to our mainly focused synchronous ADP is the so-called iterative ADP \cite{liu2013policy,sokolov2015complete,al2019online}, which sequentially updates actor and critic ANNs (i.e., one ANN is tuned, and the other holds constant). Although this method enjoys the model-free property for discrete-time systems, however, its extension to continuous-time systems entails challenges in proving system stability and ensuring that the algorithm is online and model-free \cite{bhasin2013novel}.

Among the aforementioned robust synchronous ADP related works, either complete \cite{vamvoudakis2012online,liu2015reinforcement} or partial model knowledge \cite{vrabie2009neural} is required.
The desired model-free control is accomplished by introducing additional techniques, e.g., NNs \cite{bhasin2013novel,zhao2019event,zhang2011data}, fuzzy models \cite{su2019adaptive}, GP \cite{boedecker2014approximate}, or observers \cite{sun2018disturbance}, where the dynamics is required to be identified online explicitly.  
Unlike these computation-intensive approaches, time delay estimation (TDE) \cite{hsia1990robot,youcef1992input} is a fundamentally different mechanism to facilitate model-free control strategies, 
where an incremental dynamics constructed by time-delayed signals is used to reflect the system evolution of the controlled system incrementally without introducing any online identification processes. 
However, despite its promising robustness feature and beneficial computation simplicity, the optimality property of TDE based methods remains to be investigated.
The implementation of TDE unintentionally introduces the TDE error, which denotes the gap between the real system and the constructed incremental dynamics.  
Although the boundness property of this TDE error is analyzed in traditional TDE related works \cite{hsia1990robot,youcef1992input}, its influence on the controller performance is overlooked.
A fundamental problem about addressing the TDE error has yet to be properly established. 
The idea of exploiting sensory data to facilitate the model-free approximate optimal control strategy originates in \cite{zhou2016nonlinear,zhou2018incremental,zhou2020incremental} where the Taylor series expansion based incremental control technique is used to reduce dependence on the explicit knowledge of dynamics. 
However, no system stability is presented in related works \cite{zhou2016nonlinear,zhou2018incremental,zhou2020incremental}.
Besides, although this method avoids identifying a global system model, a recursive least square method is still required to identify local system transition matrices, which introduces additional computation load.

This paper proposes a fundamentally different approach to achieve model-free control with guaranteed stability and optimality. 
This is accomplished by first leveraging TDE to get an equivalent incremental dynamics to the investigated dynamics, which alleviates the need for the online identification process, as well as its accompanying computation complexity and parameter tuning efforts.
Then, the resulting incremental dynamics serves as a basis to allow the design of the
model-free approximate optimal incremental control strategy.
The contribution of this work is summarized as follows.
\begin{enumerate}
\item We develop a novel RL augmented control approach, which is called IADP, that enjoys both model-free feature and guaranteed closed-loop system stability.
More importantly, IADP accomplishes a significant reduction in the control energy expenditure, which enables it to be favourable to energy-limited platforms.
\item Under an optimization framework, performance indexes regarding state deviations and control energy expenditure are considered. Thus, we endow TDE based methods with the optimality property. Besides, by incorporating a TDE error bound related term into the cost function, we novelly attenuate the TDE error during the optimization process.
\end{enumerate}

The remainder of this paper is organized as follows: problem formulation of the robust stabilization problem, problem transformation to the optimal incremental control problem, and problem equivalence proofs are provided in Section \ref{Section problem formu}. Thereafter, we present the approximate optimal solution in Section \ref{sect NN update law}. Numerical simulation results shown in Section \ref{sec simulation} demonstrate the effectiveness and superiority of IADP. Finally, Section \ref{sec conclusion} concludes this work.

\emph{Notations}
Throughout this paper, $\mathbb{R}$ ($\mathbb{R}^{+}$) denotes the set of real (positive) numbers; $\mathbb{R}^{n}$ is the Euclidean space of $n$-dimensional real vector; $\mathbb{R}^{n \times m}$ is the Euclidean space of $n \times m$ real matrices; $I_{m \times m}$ represents the identity matrix with dimension $m \times m$; 
$\lambda_{\min}(M)$ and $\lambda_{\max}(M)$ are the maximum and minimum eigenvalues of a symmetric matrix $M$, respectively; 
The $i$-th entry of a vector $x = [x_{1},...,x_{n}]^{\top}\in \mathbb{R}^{n}$ is denoted by $x_{i}$, and $\left\| x \right\| = \sqrt{\sum_{i=1}^{N}|x_{i}|^2}$ is the Euclidean norm of the vector $x$;
The $ij$-th entry of a matrix $D \in \mathbb{R}^{n \times m}$ is denoted by $d_{ij}$, and $\left\|D\right\| = \sqrt{\sum_{i=1}^{n}\sum_{j=1}^{m}|d_{ij}|^2}$ is the Frobenius norm of the matrix $D$. For notational brevity, time-dependence is suppressed without causing ambiguity. 
\section{Problem formulation} \label{Section problem formu}
Considering the following continuous-time nonlinear system:
\begin{equation}\label{sys}
\dot{x} = f(x) + g(x)u(x) +d(t),
\end{equation}
where $x \in \mathbb{R}^{n}$, $u(x) \in \mathbb{R}^{m}$ are system states and inputs, respectively. 
$f(x) : \mathbb{R}^{n} \to \mathbb{R}^{n}$, $g(x) : \mathbb{R}^{n} \to \mathbb{R}^{n \times m}$ are continuous and locally Lipschitz drift and input dynamics, respectively.
$d(t) \in \mathbb{R}^{n}$ represents a bounded time-varying external disturbance.
Assuming that no knowledge of the dynamics \eqref{sys} is available expect for the dimension of system states and inputs.

The main objective of this paper is to tackle the robust stabilization problem of the highly uncertain dynamics \eqref{sys} that operates in a disturbed environment, which is formulated as Problem \ref{Robust stabilization}.
\begin{problem} \label{Robust stabilization}
    Design a control strategy $u(x)$ such that the system \eqref{sys} is stable under input saturation
    $\mathbb{U}_j = \left\{u_j \in \mathbb{R} : \left| u_{j} \right| \leq \beta \right\}, j = 1, \cdots, m$, where $\beta \in \mathbb{R}^{+}$ is a known saturation bound.
\end{problem}
\begin{remark}
Note that although the explicit form of the controlled system \eqref{sys} is provided here, which is introduced for the analytical purpose and facilitates the controller design as well as the stability analysis in the following sections, our developed control approach relies on neither model parameters nor environmental information.
\end{remark}

\subsection{Incremental dynamics} \label{Sec incrmental dyna}
The highly uncertain dynamics \eqref{sys} cannot be directly used to design a controller to solve Problem \ref{Robust stabilization}. Therefore, based on measured input-state data, this section leverages the TDE technique to get an incremental dynamics that is an equivalent of \eqref{sys}. 
This formulated incremental dynamics reflects the system response of the controlled system \eqref{sys} incrementally without using any explicit model parameters, or any preceding identification procedures.
Here, the attempt to relieve dependence on the accurate knowledge of dynamics departs from existing works where additional computation-intensive tools such as NNs \cite{bhasin2013novel,zhao2019event,zhang2011data}, fuzzy models \cite{su2019adaptive}, GP \cite{boedecker2014approximate}, or observers \cite{sun2018disturbance} are required to address model uncertainties and/or environmental disturbances. The constructed incremental dynamics in this section serves as a basis for the development of the desired model-free control strategy and the rigorous closed-loop system stability analysis in the following sections.

Before proceeding, the following assumption is provided to facilitate the formulation of an incremental dynamics.
\begin{assumption} \cite{kamalapurkar2015approximate}\label{asp of g}
 The input dynamics $g(x)$ is bounded, the matrix $g(x)$ has full column rank $\forall x \in \mathbb{R}^n$, and $g^{+}=(g^{\top}g)^{-1}g^{\top} \in \mathbb{R}^{m \times n}$ is bounded and Lipschitz continuous. 
\end{assumption}
\begin{remark}
Assumption \ref{asp of g} is not restrictive and common in ADP related works \cite{kamalapurkar2015approximate,kiumarsi2014actor}. Here,
the introduced $g^{+}$ is used to extend the TDE method usually applied to the Euler-Lagrange equation \cite{hsia1990robot,youcef1992input} to the general nonlinear system \eqref{sys}.
\end{remark}

To get the incremental dynamics, we start with introducing a constant matrix $\bar{g} \in \mathbb{R}^{n \times m}$ and multiply $\bar{g}^{+}$ on the dynamics \eqref{sys},
\begin{equation}\label{sys gbar}
\bar{g}^{+} \dot{x} 
= \bar{g}^{+} f(x)  + \bar{g}^{+}g(x)u(x) +\bar{g}^{+} d(t)
= H(x,\dot{x},u(x)) + u(x),
\end{equation}
where  $H(x,\dot{x},u(x))  = (\bar{g}^{+}-g^{+}(x))\dot{x}+g^{+}(x)f(x)+g^{+}(x)d(t) : \mathbb{R}^{n} \times \mathbb{R}^{n} \times \mathbb{R}^{m} \to \mathbb{R}^{m}$. It is a lump term that embodies all the unknown model knowledge (i.e., $f(x)$, $g(x)$) as well as external disturbances (i.e., $d(t)$).
Based on the TDE technique \cite{hsia1990robot,youcef1992input}, with a sufficiently high sampling rate, the unknown $H(x,\dot{x},u(x))$ in \eqref{sys gbar} could be estimated by time-delayed signals as 
\begin{equation}\label{TDE H}
\hat{H}(x,\dot{x},u(x)) = H(x_0,\dot{x}_0,u_0) = (\bar{g}^{+}-g^{+}_0)\dot{x}_0+g^{+}_0f_0+g^{+}_0 d_0, 
\end{equation}
where ${x}_0 = x(t-L)$, $u_0 = u(x(t-L))$, $g_0 = g(x(t-L))$, $f_0 = f(x(t-L))$, and $d_0 = d(t-L)$ with $L \in \mathbb{R}^+$. Specifically, the TDE of $H(x,\dot{x},u(x))$ is represented by the measured data at time $t-L$. 
\begin{remark}
In practice, a digital control system behaves reasonably close to a continuous system if the sampling rate is faster than $30$ times the system bandwidth \cite{franklin1998digital}. Therefore, with a sufficiently high sampling rate, based on TDE, the unknown lumped nonlinear function $H(x,\dot{x},u(x))$ could be estimated by reusing past measured input-state data.
\end{remark}
Then, by substituting \eqref{TDE H} into \eqref{sys gbar}, we get 
\begin{eqnarray}\label{sys gbar Hbar}
\bar{g}^{+} \dot{x} 
=  \hat{H}(x,\dot{x},u(x)) + u(x) + H(x,\dot{x},u(x))-\hat{H}(x,\dot{x},u(x))
= H(x_0,\dot{x}_0,u_0) + u(x) + \xi,
\end{eqnarray}
where $\xi = H(x,\dot{x},u(x))-\hat{H}(x,\dot{x},u(x)) \in \mathbb{R}^{m}$ denotes the so-called TDE error, which is proved to be bounded as given in Lemma \ref{boud of TDE AE}. Comparing \eqref{sys gbar} with \eqref{sys gbar Hbar}, we know that the unknown lumped term $H(x,\dot{x},u(x))$ has been reconstructed from the measured data based $H(x_0,\dot{x}_0,u_0)$ plus the TDE error $\xi$.

In addition, at time $t-L$, \eqref{sys gbar} follows here
\begin{equation}\label{sys original step}
    \bar{g}^{+} \dot{x}_0 = H(x_0,\dot{x}_0,u_0) + u_0.
\end{equation}

Finally, subtracting \eqref{sys original step} from \eqref{sys gbar Hbar},  we get the incremental dynamics as
\begin{equation}\label{sys incrmental form}
    \Delta \dot{x} = \bar{g}\Delta u + \bar{g}\xi,
\end{equation}
where $\Delta \dot{x} = \dot{x} -\dot{x}_0 \in \mathbb{R}^{n}$, $\Delta u = u(x)-u_0 \in \mathbb{R}^{m}$.
Here, with a predefined $\bar{g}$, measured input-state data are adopted to reflect the system response in an incremental way without using any model or environmental information.

\begin{remark}
The TDE technique, which is usually used in the robotic filed \cite{hsia1990robot,youcef1992input}, is extended to the general continuous-time nonlinear system in this section. 
From a practical perspective, the applied TDE technique enables us to switch from the requirement of accurate mathematical models to sensor capabilities of providing accurate measurements of $\Delta \dot{x}$ (constructed from $\dot{x}$ and $\dot{x}_0$) and $\Delta u$ (constructed from $u(x)$ and $u_0$). For cases where $\dot{x}$ and $\dot{x}_0$ are not directly measurable,
multiple state derivative estimation techniques developed in previous works \cite{bhasin2012robust, levant1998robust} could help, which is out of the scope of this paper and remains as future works.
\end{remark}
However, although an equivalent of \eqref{sys} is provided in \eqref{sys incrmental form} without using any explicit knowledge of dynamics, the unknown TDE error $\xi$ hinders us to directly utilize \eqref{sys incrmental form} to design controllers.
Therefore, a method will be developed to address the TDE error $\xi$ in the next section. Before proceeding, here we first provide the theoretical analysis about the boundness property of $\xi$, which facilitates the method to tackle the TDE error $\xi$ under an optimization framework in Section \ref{Sec problem formulation}.
\begin{lemma} \label{boud of TDE AE}
Given a sufficiently high sampling rate, $\exists \Bar{\xi} \in \mathbb{R}^+$, there holds $\left\|\xi\right\| \leq \Bar{\xi}$. 
\end{lemma}
\begin{proof}
Combining \eqref{sys gbar} with \eqref{TDE H}, the TDE  error follows
\begin{eqnarray}\label{AE bound 1}
\xi 
&=& H(x,\dot{x},u(x))-\hat{H}(x,\dot{x},u(x)) = H(x,\dot{x},u(x)) - H(x_0,\dot{x}_0,u_0)\nonumber \\
&=& (\bar{g}^{+}-g^{+}(x))(\dot{x}-\dot{x}_0) + (g^{+}_0 - g^{+}(x)) \dot{x}_0 + g^{+}(x)f(x)-g^{+}_0 f_0 \nonumber
+ g^{+}(x)d(t)-g^{+}_0 d_0 \nonumber \\
&=&  (\bar{g}^{+}-g^{+}(x)) \Delta \dot{x} + (g^{+}_0 - g^{+}(x)) \dot{x}_0 + g^{+}(x) (f(x)-f_0)+ (g^{+}(x)-g^{+}_0)f_0+g^{+}(x)(d(t)-d_0)+(g^{+}(x)-g^{+}_0) d_0.
\end{eqnarray}
Besides, based on the system \eqref{sys}, we get 
\begin{eqnarray}\label{AE bound 2}
\Delta \dot{x}
&=& f(x)+g(x)u(x)+d(t)-f_0-g_0u_0-d_0 \nonumber \\
&=& g(x)  \Delta u + (g(x)-g_0)u_0+  f(x) - f_0 + d(t)-d_0.
\end{eqnarray}
Then, substituting \eqref{AE bound 2} into \eqref{AE bound 1} yields
\begin{eqnarray}\label{AE bound 3}
\xi 
&=& (\bar{g}^{+}-g^{+}(x)) g(x)  \Delta u  +(\bar{g}^{+}-g^{+}(x)) [(g(x)-g_0)u_0+  f(x) - f_0 + d(t)-d_0]+(g^{+}_0 - g^{+}(x)) \dot{x}_0 \nonumber \\
&+& g^{+}(x) (f(x)-f_0)+ (g^{+}(x)-g^{+}_0)f_0+g^{+}(x)(d(t)-d_0)+(g^{+}(x)-g^{+}_0) d_0 \nonumber \\
&=& (\bar{g}^{+}g(x)-I_{m \times m})\Delta u+\delta_1 
\leq  \left\| \bar{g}^{+}g(x)-I_{m \times m} \right\| \left\| \Delta u\right\| +\left\| \delta_1\right\|,
\end{eqnarray}
where $\delta_1 = \bar{g}^{+}(g(x)-g_0)u_0+ \bar{g}^{+}(f(x)-f_0) + \bar{g}^{+}(d(t)-d_0) $.

For a sufficiently high sampling rate, the gap between successive states is sufficiently small.
Thus, it is reasonable to assume that there exists a positive constant 
$\bar{\delta}_1 \in \mathbb{R}^+$ such that $\left\| \delta_1 \right\|\leq \bar{\delta}_1$. In addition, the bounded control input $u$ implies that $\left\| \Delta u\right\|\leq 2 \beta$ holds. By choosing a suitable $\bar{g}$ such that $\left\| \bar{g}^{+}g(x)-I_{m \times m} \right\|\leq c$ establishes, then we get
\begin{equation} \label{TDE AE bound final}
    \xi \leq c \left\| \Delta u\right\|+ \bar{\delta}_1 \leq 2\beta c+ \bar{\delta}_1 = \bar{\xi}.
\end{equation}
This concludes the proof.
\end{proof}

\begin{remark}
By using the Taylor series expansion based incremental control technique, previous works \cite{zhou2016nonlinear,zhou2018incremental,zhou2020incremental,acquatella2013incremental,simplicio2013acceleration} attempt to provide the incremental dynamics by offering the first-order approximation of $\dot{x}$ in the neighbourhood of $[x_0,u_0]$. It follows
\begin{eqnarray}\label{Taylor series}
\dot{x} 
&=& f(x) + g(x) u(x) \nonumber\\
&=& f_0 + g_0 u_0 + \frac{\partial [f(x)+g(x)u(x)]}{\partial x}|_{x = x_0, u=u_0} (x-x_0) + \frac{\partial [f(x)+g(x)u(x)]}{\partial u}|_{x = x_0, u=u_0} (u-u_0) + \mathcal{H.O.T}. \nonumber\\
& \cong & \dot{x}_0 + F[x_0,u_0] \Delta x + G[x_0,u_0] \Delta u \nonumber,
\end{eqnarray}
where $F[x_0,u_0] = [\partial (f(x)+g(x)u(x))/\partial x]|_{x = x_0, u=u_0} \in \mathbb{R}^{n\times n}$ is the system matrix, and $G[x_0,u_0] = [\partial (f(x)+g(x)u(x))/\partial u]|_{x = x_0, u=u_0} \in \mathbb{R}^{n\times m}$ is the control effectiveness matrix.
However, the approximation error resulting from the high order term $\mathcal{H.O.T}$ is directly omitted without considering its influence on the controller performance. Furthermore, a recursive least square method is demanded to search for suitable gain matrices $F[x_0,u_0]$ and $G[x_0,u_0]$ to construct the incremental dynamics \cite{zhou2016nonlinear,zhou2018incremental,zhou2020incremental}. This required online identification of $F[x_0,u_0]$ and $G[x_0,u_0]$ introduces additional computational burden.
\end{remark}
\subsection{Problem transformation to optimal incremental control} \label{Sec problem formulation}
To address the unknown TDE error in the incremental dynamics \eqref{sys incrmental form}, here we attempt to investigate the original robust stabilization problem shown as Problem \ref{Robust stabilization} from an optimal control perspective, whereby the TDE error could be reflected in the performance index and further be attenuated during the optimization process. This departs from existing TDE related works \cite{hsia1990robot,youcef1992input,zhou2016nonlinear,zhou2018incremental,zhou2020incremental,acquatella2013incremental,simplicio2013acceleration} that directly ignore the influence of the TDE error on the controller performance. 
Moreover, the effort to solve Problem \ref{Robust stabilization} under an optimization framework enables us to take the desired performance indexes regarding state deviations and control energy expenditure into consideration. These considered performance indexes endow the resulting TDE based model-free control strategy with guaranteed optimality.

The TDE error $\xi$ in \eqref{sys incrmental form} is unknown. Thus, the available incremental dynamics to design a controller to solve Problem \ref{Robust stabilization} follows
\begin{equation}\label{sys incrmental form without epsilon}
   \Delta \dot{x} = \bar{g}\Delta u.
\end{equation}
To attenuate the TDE error $\xi$ that is overlooked in \eqref{sys incrmental form without epsilon}, as well as to optimize the performance of states and control inputs,
we consider the cost function of \eqref{sys incrmental form without epsilon} as
\begin{equation}\label{cost fuction}
    V(x(t)) = \int_{t}^{\infty} r(x(\tau),\Delta u(\tau))\,d\tau,
\end{equation}
where $r(x,\Delta u) = x^{\top} Q x + \mathcal{W}(u_0+\Delta u) + \Bar{\xi}^2_o$. 
The common squared term $x^{\top} Q x$ reflects users' preference for the controller performance concerning state deviations, where $Q \in \mathbb{R}^{n \times n}$ is a positive definite matrix.
The non-squared control penalty function $\mathcal{W}(u_0+\Delta u)$, which relates to the measured $u_0$ and to be designed $\Delta u$, is introduced to enforce the control limit on $u(x)$ based on the bounded $\tanh$ function. The explicit form of this part follows \cite{abu2005nearly}
\begin{equation}\label{control cost function}
    \mathcal{W}(u_0+\Delta u) =  2 \sum_{j=1}^{m}\int_{0}^{u_{0_{j}}+\Delta u_j}  \beta \tanh^{-1}(\vartheta_j / \beta)  \,d\vartheta_j.
\end{equation}
Originally, we could incorporate the squared TDE error bound $\Bar{\xi}^2$ into $r(x,\Delta u)$ to attenuate the TDE error $\xi$ during the optimization process. However, according to \eqref{TDE AE bound final} of Lemma \ref{boud of TDE AE}, the explicit value of $\Bar{\xi}$ is unknown. 
Thus, we seek for $\Bar{\xi}^2_o$, where $\Bar{\xi}_o = \bar{c}  \left\| \Delta u\right\|$ and  $\bar{c} \in \mathbb{R}^+$ is chosen as illustrated in Theorem \ref{theorem equalivance}, to replace $\Bar{\xi}^2$ to accomplish the same goal.
It is worth noting that the designed utility function $r(x,\Delta u)$ here enables us to perform the optimization of incremental control inputs. This achievable optimization of  incremental control inputs enables IADP to enjoy control effort reductions, which will be verified in the simulation part. This is one distinguishing feature of our proposed IADP.

 \begin{remark}
 Note that there exist other options to address the TDE error $\xi$. For example, by treating the unknown TDE error $\xi$ in \eqref{sys incrmental form} as a kind of disturbance, we can introduce the widely used disturbance-observer based methods \cite{chen2015disturbance} or sliding mode control methods \cite{shtessel2014sliding} to compensate the TDE error $\xi$. Comparing to these add-on methods, our strategy enjoys computation simplicity.
 \end{remark}

The aforementioned settings allow us to formulate an optimal incremental control problem in Problem \ref{Optimization regulation}, whose equivalence to Problem \ref{Robust stabilization} will be later proved in Theorem \ref{theorem equalivance}.
\begin{problem} \label{Optimization regulation}
    Given Assumption \ref{asp of g} and Lemma \ref{boud of TDE AE}, 
    consider the incremental dynamics \eqref{sys incrmental form without epsilon}, find an incremental control strategy $\Delta u$ to minimize the cost function defined as \eqref{cost fuction}.
\end{problem}
Before proceeding to solve Problem \ref{Optimization regulation}, by following \cite[Definition 1]{abu2005nearly} where admissible controls are defined with respect to \eqref{sys}, here we provide the corresponding definition of admissible incremental controls concerning \eqref{sys incrmental form without epsilon}, which facilitates the derivation of the closed-form optimal incremental control strategy.
\begin{definition}[Admissible incremental control]\label{def of admissible incremental control}
An incremental control $\Delta u$ is defined to be admissible with respect to \eqref{sys incrmental form without epsilon} on $\Omega \subseteq \mathbb{R}^{n}$, denoted by $\Delta u \in \Psi (\Omega)$, if $\Delta u$ is continuous on $\Omega$, $\Delta u(0)=0$, $\Delta u$ stabilizes \eqref{sys incrmental form without epsilon} on $\Omega$, and \eqref{cost fuction} is finite.
\end{definition}

Aiming at Problem \ref{Optimization regulation}, for any admissible incremental control policies $\Delta u \in \Psi (\Omega)$, the associated optimal cost function follows
\begin{equation} \label{optimal cost function}
    V^{*}(x(t)) = \min_{\Delta u}\int_{t}^{\infty} r(x(\tau),\Delta u(\tau))\,d\tau.
\end{equation}
Then, the HJB equation follows
\begin{equation} \label{HJB equation}
    H(x,\Delta u^{*},\nabla V^{*}) = r(x,\Delta u^{*})
    + \nabla V^{{*}^{T}}(\Delta \dot{x}+\dot{x}_0) = r(x,\Delta u^{*})
    + \nabla V^{{*}^{T}}(\bar{g} \Delta u^*+\dot{x}_0 )=0,
\end{equation}
where the operator $\nabla$ denotes the partial derivative with regard to $x$, i.e., $\partial (\cdot) / \partial x $.
    
Assuming that the minimum of \eqref{optimal cost function} exits and is unique \cite{vamvoudakis2010online}. By using the stationary optimality condition on the HJB equation \eqref{HJB equation}, i.e., $\partial H(x,\Delta u^{*},\nabla V^{*}) / \partial \Delta u = 0$ , we get the closed-form optimal incremental control strategy as
\begin{equation} \label{optimal incremental u}
    \Delta u^{*} = - \beta \tanh(\frac{1}{2\beta} \bar{g} ^{\top}\nabla V^{*})-u_0.
\end{equation}
Then, we could construct the corresponding optimal control strategy as 
\begin{equation} \label{optimal u}
    u^{*} = u_0 + \Delta u^{*} = - \beta \tanh(\frac{1}{2\beta} \bar{g}^{\top}\nabla V^{*}). 
\end{equation}
Departing from traditional ADP related works \cite{vamvoudakis2010online,kamalapurkar2015approximate} where the total optimal control input $u^{*}$ is directly designed, here we first get the theoretically derived incremental optimal control strategy $\Delta u^*$ in \eqref{optimal incremental u}, and then construct $u^{*}$ based on the measured $u_0$ and the designed $\Delta u^*$.
This difference lies in that Problem \ref{Optimization regulation} is formulated based on the incremental dynamics \eqref{sys incrmental form without epsilon} that relates with incremental states and control inputs.  
\begin{remark}
Alternatively, we could replace $\mathcal{W}(u_0+\Delta u)$ in $r(x,\Delta u)$ with $\mathcal{W}(\Delta u)  =  2 \sum_{j=1}^{m}\int_{0}^{\Delta u_j}  \alpha \tanh^{-1}(\vartheta_j / \alpha)  \,d\vartheta_j$. This enforces the constraint satisfaction of the incremental control inputs, which is denoted as $- \alpha \leq \Delta u_j \leq \alpha$,  $\alpha \in \mathbb{R}^{+}$, $j=1,\cdots,m$. By following the aforementioned derivation processes \eqref{optimal cost function}-\eqref{optimal u}, the corresponding optimal incremental control follows $\Delta u^{*} = - \alpha \tanh(\frac{1}{2\alpha} \bar{g} ^{\top}\nabla V^{*})$. Then, the resulting optimal control is $u^{*} = u_0+\Delta u^{*}$. However, in this case, the control limit on $u(x)$ cannot be addressed. Given that input saturation is common in real life and violations of it might lead to serious consequences, we prefer to incorporate \eqref{control cost function} into $r(x,\Delta u)$ to enforce the control limit on $u(x)$.
\end{remark}
To get $\Delta u^{*}$ \eqref{optimal incremental u} and $u^{*}$ \eqref{optimal u},
$\nabla V^{*}$ remains to be determined.
 We defer the explicit method to acquire $\nabla V^{*}$ in Section \ref{sect NN update law}, and focus now on the equivalence proof to show that after solving Problem \ref{Optimization regulation},
the resulting $u^*$ \eqref{optimal u} constructed from the designed $\Delta u^*$ \eqref{optimal incremental u} is the robust stabilization solution to Problem \ref{Robust stabilization}.
\begin{theorem}\label{theorem equalivance}
Given Assumption \ref{asp of g} and Lemma \ref{boud of TDE AE}, consider the system described by \eqref{sys}, the optimal control strategy \eqref{optimal u} constructed from the optimal incremental control strategy \eqref{optimal incremental u} guarantees robust stabilization of \eqref{sys}, if there exists a scalar $\bar{c} \in \mathbb{R}^{+}$ such that the following inequality is satisfied
     \begin{equation} \label{equalivance condition}
    \bar{\xi} < \bar{c}  \left\| \Delta u\right\|.
    \end{equation}
\end{theorem}
\begin{proof}
Given that $V^{*}(x=0) =0$,  and $V^{*} >0$ for $\forall x \ne 0$,  $V^{*}$ defined in \eqref{optimal cost function} could serve as a Lyapunov function candidate for the stability proof. Taking time derivative of $V^{*}$ along the incremental dynamics \eqref{sys incrmental form}, which is an equivalent of the original dynamics \eqref{sys}, it follows
\begin{equation} \label{equalivance 1}
     \dot{V}^{*} = \nabla {V^*}^{\top}(\Delta \dot{x}+\dot{x}_0)
     =\nabla {V^*}^{\top}( \bar{g} \Delta u^* + \bar{g}\xi+\dot{x}_0)
                 = \nabla {V^*}^{\top}( \bar{g} \Delta u^* +\dot{x}_0 ) + \nabla {V^*}^{\top}\bar{g}\xi.
\end{equation}
According to \eqref{HJB equation} and \eqref{optimal incremental u}, we get
\begin{equation} \label{equalivance 2}
\nabla {V^*}^{\top}( \bar{g} \Delta u^* +\dot{x}_0 ) = -x^{\top} Q x - \mathcal{W}(u_0+\Delta u^{*}) - \Bar{\xi}^2_o, \ \nabla {V^*}^{\top}\bar{g} = - 2 \beta \tanh^{-1}(\frac{u_0 + \Delta u^{*}}{\beta}).
\end{equation}
Substituting \eqref{equalivance 2} into \eqref{equalivance 1} reads
\begin{eqnarray} \label{equalivance 3}
     \dot{V}^{*} 
     &=& -x^{\top} Q x - \mathcal{W}(u_0+\Delta u^{*}) - \Bar{\xi}^2_o - 2 \beta \tanh^{-1}(\frac{u_0 + \Delta u^{*}}{\beta}) \xi.
\end{eqnarray}
As for $\mathcal{W}(u_0+\Delta u^{*})$ in \eqref{equalivance 3}, based on the explicit form in \eqref{control cost function} and by setting $\varsigma_j = \tanh^{-1}(\vartheta_j / \beta)$, it follows
    \begin{eqnarray} \label{equalivance 4}
    \mathcal{W}(u_0+\Delta u^{*}) &=& 2 \beta \sum_{j=1}^{m}\int_{0}^{u_{0_{j}}+\Delta u^{*}_{j} }  \tanh^{-1}(\vartheta_j / \beta)\,d\vartheta_j 
    = 2 \beta^2 \sum_{j=1}^{m}\int_{0}^{\tanh^{-1}(\frac{u_{0_{j}}+\Delta u^{*}_{j}}{\beta})} \varsigma_j (1-\tanh^2 (\varsigma_j))\,d \varsigma_j \nonumber \\
                     &=& \beta^2 \sum_{j=1}^{m} (\tanh^{-1}(\frac{u_{0_{j}}+\Delta u^{*}_{j}}{\beta}))^2 - \epsilon_{u},
    \end{eqnarray}
where $\epsilon_{u} = 2 \beta^2 \sum_{j=1}^{m} \int_{0}^{\tanh^{-1}(\frac{u_{0_{j}}+\Delta u^{*}_{j}}{\beta} )} \varsigma_j \tanh^2(\varsigma_j)\,d\varsigma_j$. Based on the integral mean-value theorem, there exists a series of $\theta_j \in [0,\tanh^{-1}(\frac{u_{0_{j}}+\Delta u^{*}_{j}}{\beta})], j = 1, \cdots, m$, such that
\begin{equation} \label{epsilon_u_1}
    \epsilon_{u} =  2 \beta^2 \sum_{j=1}^{m} \tanh^{-1}(\frac{u_{0_{j}}+\Delta u^{*}_{j}}{\beta}) \theta_j \tanh^2(\theta_j).
\end{equation}
Based on \eqref{equalivance 2} and the fact  $0\leq \tanh^2(\theta_j)\leq 1$, it follows
    \begin{equation} \label{epsilon_u_2}
            \epsilon_{u}  \leq  2 \beta^2 \sum_{j=1}^{m}  \tanh^{-1}(\frac{u_{0_{j}}+\Delta u^{*}_{j}}{\beta} )\theta_j 
             \leq 2 \beta^2 \sum_{j=1}^{m }(\tanh^{-1} (\frac{u_{0_{j}}+\Delta u^{*}_{j}}{\beta}))^2
             = \frac{1}{2} \nabla {V^*}^{\top}\bar{g} \bar{g}^{\top}\nabla V^{*}.
    \end{equation}
    The definition of admissible incremental control in Definition \ref{def of admissible incremental control} concludes that $V^*$ is finite. Additionally, there exists $b_{\nabla V^*} \in \mathbb{R}^+$ such that $\left\|\nabla V^{*}\right\| \leq b_{\nabla V^*}$. Thus, we could rewrite \eqref{epsilon_u_2} as
    \begin{equation} \label{epsilon_u_3}
            \epsilon_{u}  \leq b_{\epsilon_{u}} = \frac{1}{2}  \left\|\bar{g}\right\|^2 b^2_{\nabla V^*}.
    \end{equation}
Then, substituting \eqref{equalivance 4}, \eqref{epsilon_u_3} into \eqref{equalivance 3} yields
\begin{eqnarray} \label{equalivance 5}
     \dot{V}^{*} 
     \leq -x^{\top} Q x - (\Bar{\xi}^2_o-\xi^{\top}\xi ) - [\beta \tanh^{-1}(\frac{u_0 + \Delta u^{*}}{\beta})+\xi]^2+ b_{\epsilon_{u}}. 
\end{eqnarray}
By choosing $\Bar{\xi}_o = \bar{c}  \left\| \Delta u\right\|$, and $\bar{c}$ is chosen to satisfy $ \bar{c}  \left\| \Delta u\right\| > \bar{\xi}$, where $\bar{\xi}$ is defined in \eqref{TDE AE bound final}, the following inequality holds
\begin{eqnarray} \label{equalivance 5}
     \dot{V}^{*} 
     \leq -x^{\top} Q x + b_{\epsilon_{u}}.
\end{eqnarray}
Thus, $\dot{V}^{*} < 0$ holds if $-\lambda_{\min}(Q) \left\|x\right\|^2+b_{\epsilon_{u}} < 0.$ Finally, it concludes that states converge to the residual set 
\begin{equation} \label{equalivance 6}
\Omega_{x} = \{x | \left\| x \right\| \leq \sqrt{b_{\epsilon_{u}}/\lambda_{\min}(Q)} \}.
\end{equation}
The aforementioned proof means that based on the optimal cost function \eqref{optimal cost function},  the derived optimal incremental control policy \eqref{optimal incremental u} of the system \eqref{sys incrmental form without epsilon}  robustly stabilizes the system \eqref{sys incrmental form}. Given the equivalence between \eqref{sys} and \eqref{sys incrmental form} clarified in Section \ref{Sec incrmental dyna}, thus the optimal control input \eqref{optimal u}, which is constructed from the designed \eqref{optimal incremental u}, robustly stabilizes the system \eqref{sys}.
This concludes the proof.
\end{proof}
We have proved in Theorem \ref{theorem equalivance} that the optimal incremental control problem clarified in Problem \ref{Optimization regulation} is equivalent to the robust stabilization problem shown as Problem \ref{Robust stabilization}. Thus, to stabilize the high uncertain dynamics \eqref{sys} operating in a disturbed environment, the following paper devotes to solving Problem \ref{Optimization regulation}.
\section{Approximate optimal solution}\label{sect NN update law}
To solve Problem \ref{Optimization regulation}, this section seeks for the approximate solution to the value function of the HJB equation \eqref{HJB equation} that is hard to solve directly. 
Departing from common ADP related works \cite{vamvoudakis2010online,kamalapurkar2015approximate} in an actor-critic structure, we introduce a single critic structure here, which decreases the computational burden and simplifies the theoretical analysis. In addition, the adopted critic ANN for approximating the value function is in essence a linear approximator. This enables us to transform the critic ANN weight learning problem into a parameter identification problem. Then, by further exploiting the collected experience data to provide the sufficient excitation required for the weight convergence, we design a simple yet efficient off-policy weight update law with guaranteed weight convergence. Our approach is favourable to practical applications comparing to common methods that often directly add external noises to control inputs to satisfy the persistence of excitation (PE) condition required for the weight convergence \cite{vamvoudakis2010online,tao2003adaptive}, which results in undesirable oscillations and additional control efforts.

\subsection{Value function approximation}
Based on the ANN approximation scheme, the value function is approximated as \cite{vamvoudakis2010online}
\begin{equation}\label{optimal V approximation}
    V^{*}(x) = {W^*}^{\top} \Phi(x) + \epsilon(x),
\end{equation}
where ${W^*} \in \mathbb{R}^{N}$ is a weighting matrix, $\Phi(x) : \mathbb{R}^{n} \to \mathbb{R}^{N}$ represents the activation function, and $\epsilon(x) \in \mathbb{R}$ denotes the approximation error. 
The corresponding partial derivative of  $V^{*}(x)$ follows
\begin{equation}\label{optimal dV approximation}
    \nabla V^{*}(x) = \nabla \Phi^{\top}(x)W^{*}+\nabla\epsilon(x),
    \end{equation}
    where $\nabla \Phi \in \mathbb{R}^{N \times n}$, $\nabla\epsilon(x) \in \mathbb{R}^{n}$.
     As $N \to \infty$, both $\epsilon(x)$ and $\nabla\epsilon(x)$ converge to zero uniformly. Without loss of generality, the following assumption is given, which is common in ADP related works.
     \begin{assumption} \cite{vamvoudakis2010online}\label{bound of NN issues}
        There exist known constants $b_{\epsilon}, b_{\epsilon x}, b_\Phi, b_{\Phi x}\in \mathbb{R}^{+}$ such that $\left\| \epsilon(x)  \right\| \leq b_{\epsilon}$, $\left\| \nabla\epsilon(x)  \right\| \leq b_{\epsilon x}$, $\left\| \Phi(x) \right\| \leq b_\Phi$, and $\left\| \nabla\Phi(x) \right\| \leq b_{\Phi x}$.
     \end{assumption}
    Considering a fixed incremental control input $\Delta u$, inserting \eqref{optimal dV approximation} into \eqref{HJB equation} yields
    \begin{equation}\label{approximation Lyapunov equation}
        {W^*}^{\top}\nabla \Phi( \bar{g} \Delta u +\dot{x}_0 )+r(x,\Delta u) = \epsilon_{h},
    \end{equation}
    where the residual error follows $\epsilon_{h} = -\nabla \epsilon^{\top}( \bar{g} \Delta u +\dot{x}_0 ) \in \mathbb{R}$. Assuming that there exists $b_{\epsilon_{h}} \in \mathbb{R}^+$ such that $\left\| \epsilon_{h} \right\| \leq b_{\epsilon_{h}}$.
    By focusing on the ANN parameterized \eqref{approximation Lyapunov equation}, we rewrite it into the following linear in parameter (LIP) form
    \begin{equation}\label{LIP Lyapunov equation}  
    \Theta = -{W^*}^{\top}Y+\epsilon_{h},
    \end{equation}
    where $\Theta = r(x,\Delta u) \in \mathbb{R}$, and $Y = \nabla \Phi( \bar{g} \Delta u +\dot{x}_0 ) \in \mathbb{R}^{N}$.
    Given that $\Theta$ and $Y$ can be obtained from real-time data, 
    this formulated LIP form enables the learning of $W^{*}$ to be equivalent to a parameter identification problem of an LIP system from the perspective of adaptive control, where $Y$ and $W^{*}$ can be treated as the known regressor matrix and the unknown parameter vector to be determined, respectively. The applied novel transformation here allows us to design a simple weight update law with guaranteed weight convergence in the subsequent section.
    \subsection{Off-policy weight update law} \label{OPRL}
    With the ideal critic weight $W^{*}$ in \eqref{LIP Lyapunov equation} being unknown, by denoting its estimate as $\hat{W} \in \mathbb{R}^{N}$, then we get
    \begin{equation}\label{approximation LIP Lyapunov equation}
    \hat{\Theta} = -\hat{W}^{\top}Y,
    \end{equation}
    where $\hat{\Theta} \in \mathbb{R}$ is an estimate of $\Theta$. By denoting the weight estimation error as $\Tilde{W} = \hat{W}- W^{*} \in \mathbb{R}^N$, and the approximation error as $\Tilde{\Theta} = \Theta - \hat{\Theta}\in \mathbb{R}$, subtracting \eqref{approximation LIP Lyapunov equation} from \eqref{LIP Lyapunov equation} yields
    \begin{equation}\label{approximation error}
    \Tilde{\Theta} = \Tilde{W}^{\top}Y+\epsilon_{h}.
    \end{equation}
    
    To achieve $\hat{W} \to W^{*}$,  $\hat{W}$ should be updated to minimize $E = \frac{1}{2} \Tilde{\Theta}^{\top}\Tilde{\Theta}$.
    Furthermore, to guarantee the weight convergence while minimizing $E$, experience data can be exploited to provide the required sufficient excitation. Finally, a simple yet efficient off-policy weight update law of the critic agent is designed as
    \begin{equation} \label{w update law}
        \dot{\hat{W}} = - \Gamma k_c Y\Tilde{\Theta} -  \sum_{l=1}^{P} \Gamma k_{e} Y_l\Tilde{\Theta}_{l},
    \end{equation}
    where $\Gamma \in \mathbb{R}^{N \times N}$ is a constant positive definite gain matrix. $k_c, k_{e} \in \mathbb{R}^{+}$ are constant gains to balance the relative importance between current and experience data to the online learning process. $P \in \mathbb{R}^{+}$ is the volume of the experience buffers $\mathfrak{B}$ and $\mathfrak{E}$, i.e., the maximum number of recorded data points. The regressor matrix $Y_{l}\in \mathbb{R}^{N}$ and the approximation error $\Tilde{\Theta}_{l}\in \mathbb{R}$ denote the $l$-th collected data of the corresponding experience buffers $\mathfrak{B}$ and $\mathfrak{E}$, respectively.
    
    Before proceeding to the guaranteed weight convergence proof based on \eqref{w update law}, we first clarify a rank condition about the experience buffer $\mathfrak{B}$ in Assumption~\ref{rank condition}. This rank condition serves as a richness criterion of the recorded experience data and facilitates the guaranteed weight convergence analysis in Theorem \ref{Theorem weight convergence}. 
    \begin{assumption} \label{rank condition}
          Given an experience buffer $\mathfrak{B} = [Y_{1},...,Y_{P}] \in \mathbb{R}^{N \times P}$, 
          there holds $rank(\mathfrak{B}) = N$.
    \end{assumption}
    Departing from the traditional PE condition \cite{vamvoudakis2010online,tao2003adaptive}, the rank condition in Assumption~\ref{rank condition} provides an online checkable index about the data richness required for the weight convergence, which is favourable to controller designers. Assumption~\ref{rank condition} is not restrictive, which can be easily satisfied by sequentially reusing the experience data.
     
Based on the collected sufficient rich experience data illustrated in Assumption \ref{rank condition}, here we provide the guaranteed weight convergence proof based on the off-policy weight update law \eqref{w update law}.
\begin{theorem}\label{Theorem weight convergence}
Given Assumption~\ref{rank condition}, the weight learning error $\Tilde{W}$ converges to a small neighbourhood around zero.
\end{theorem}
\begin{proof} \label{proof of weight convergency}
Consider the following candidate Lyapunov function
        \begin{equation} \label{Vcl}
            V_{W} = \frac{1}{2} \Tilde{W}^{\top} \Gamma^{-1} \Tilde{W}. 
        \end{equation}
        The time derivative of $V_{W}$ follows
        \begin{equation} \label{dVcl}
            \begin{aligned}
                \dot{V}_{W} &= \Tilde{W}^{\top} \Gamma^{-1}(- \Gamma k_c Y\Tilde{\Theta} -   \sum_{l=1}^{P}\Gamma k_{e}Y_{l}\Tilde{\Theta}_{l})
                = -k_c \Tilde{W}^{\top}Y\Tilde{\Theta} -\Tilde{W}^{\top} \sum_{l=1}^{P}k_{e} Y_{l}\Tilde{\Theta}_{l}
                &\leq - \Tilde{W}^{\top} B \Tilde{W} +\Tilde{W}^{\top}\epsilon_{\Tilde{W}}.\\
            \end{aligned}
        \end{equation}
        where $B = \sum_{l=1}^{P} k_{e}Y_{l}Y^{\top}_l \in \mathbb{R}^{N \times N}$, and $\epsilon_{\Tilde{W}} =-k_c Y\epsilon_{h}-\sum_{l=1}^{P}k_{e} Y_{l}\epsilon_{h_{l}} \in \mathbb{R}^{N}$. The boundness of $Y$ and $\epsilon_{h}$ results in bounded $\epsilon_{\Tilde{W}}$. Thus, there exists $\Bar{\epsilon}_{\Tilde{W}} \in \mathbb{R}^{+}$ such that $\left\|\epsilon_{\Tilde{W}}\right\| \leq \Bar{\epsilon}_{\Tilde{W}}$. According to Assumption~\ref{rank condition}, $B$ is positive definite. Thus, \eqref{dVcl} could be rewritten as
            \begin{equation} \label{dVcl 1}
            \begin{aligned}
                \dot{V}_{W} 
                & \leq -\left\| \Tilde{W} \right\| (\lambda_{\min}(B)\left\| \Tilde{W} \right\|- \Bar{\epsilon}_{\Tilde{W}}).
            \end{aligned}
        \end{equation}
        Therefore, $\dot{V}_{W} < 0$ if $\left\| \Tilde{W} \right\| > \frac{\Bar{\epsilon}_{\Tilde{W}}}{\lambda_{\min}(B)}$. Finally, it concludes that the weight estimation error of the critic ANN will converge to the residual set
        \begin{equation} \label{compacrt set TildeW}
            \Omega_{\Tilde{W}} = \left\{\Tilde{W}| \left\| \Tilde{W} \right\| \leq \frac{\Bar{\epsilon}_{\Tilde{W}}}{\lambda_{\min}(B)}\right\}.
        \end{equation}
        This completes the proof. 
\end{proof}
With a sufficiently large $N$, $\Bar{\epsilon}_{\Tilde{W}}$ converges to zero.
Then, according to \eqref{dVcl 1}, we get $\dot{V}_{W} \leq -\lambda_{\min}(B) \left\| \Tilde{W} \right\|^2 $, i.e., $\Tilde{W} \to 0$ exponentially as $t \to \infty$. Thus, it concludes that $\hat{W}$ guarantees convergence to $W^{*}$.

Unlike common actor-critic structure based works \cite{vamvoudakis2010online,kamalapurkar2015approximate}, the guaranteed weight convergence of $\hat{W}$ to $W^*$ in Theorem \ref{Theorem weight convergence} permits us to adopt a simplified single critic structure, where the estimated critic ANN weight $\hat{W}$ could be directly used to construct the approximate optimal incremental control strategy. 
Therefore, based on the optimal incremental control strategy in \eqref{optimal incremental u}, the approximate optimal incremental control strategy follows
\begin{equation} \label{optimal incremental hat u}
    \Delta \hat{u} = - \beta \tanh(\frac{1}{2\beta}\bar{g}^{\top}\nabla \Phi^{\top}(x)\hat{W})-u_0, 
\end{equation}
Accordingly, the approximate optimal control strategy applied at the plant \eqref{sys} follows 
\begin{equation} \label{optimal  hat u}
 \hat{u} =  u_0 + \Delta \hat{u}  = - \beta \tanh(\frac{1}{2\beta}\bar{g}^{\top}\nabla \Phi^{\top}(x)\hat{W}).
\end{equation}
\begin{remark}
From a practical perspective, the designed model-free incremental control strategy \eqref{optimal incremental hat u} only requires one manually tuned constant matrix $\Bar{g}$. This feature of IADP decreases the required parameter tuning efforts comparing to existing identification based methods to fulfill model-free control strategies \cite{bhasin2013novel,zhao2019event,zhang2011data,su2019adaptive, boedecker2014approximate,sun2018disturbance}, where multiple hyperparameters or gains need to be tuned. 
\end{remark}   
Based on the off-policy weight update law \eqref{w update law}, the approximate optimal incremental control strategy \eqref{optimal incremental hat u}, and the approximate optimal control strategy \eqref{optimal  hat u}, 
we provide the main conclusions of this paper in Theorem \ref{main Theorem}.
\begin{theorem}\label{main Theorem}
Consider the incremental dynamics \eqref{sys incrmental form without epsilon}, the off-policy weight update law of the critic ANN in \eqref{w update law}, and the approximate optimal incremental control policy \eqref{optimal incremental hat u}. Given Assumptions \ref{asp of g}-\ref{rank condition}, for a sufficiently large $N$, 
the approximate optimal incremental control policy \eqref{optimal incremental hat u} stabilizes the incremental dynamics \eqref{sys incrmental form without epsilon}, and the critic ANN weight learning error $\Tilde{W}$ is uniformly ultimately bounded (UUB).
\end{theorem}
\begin{proof}
Consider the following candidate Lyapunov function
    \begin{equation} \label{Lya function stability}
    J =  V^{*}(x) + \frac{1}{2} \Tilde{W}^{\top} \Gamma^{-1} \Tilde{W}.
    \end{equation}
    By denoting  $\dot{L}_{V} = \dot{V}^{*}(x)$ and $\dot{L}_{W}= \Tilde{W}^{\top} \Gamma^{-1} \dot{\hat{W}}$, the time derivative of \eqref{Lya function stability} reads 
    \begin{equation} \label{stability dV}
       \dot{J}  = \dot{L}_{V}+\dot{L}_{W}.
    \end{equation}
    \emph{The first term $\dot{L}_{V}$ follows}
\begin{equation} \label{stability dLv 1}
            \dot{L}_{V} 
     =\nabla {V^*}^{\top}( \bar{g} \Delta \hat{u} + \bar{g}\xi+\dot{x}_0)
                 = \nabla {V^*}^{\top}( \bar{g} \Delta u^* + \dot{x}_0) +\nabla {V^*}^{\top} \bar{g}\xi+ \nabla {V^*}^{\top}\bar{g}(\Delta \hat{u}-\Delta u^*).
    \end{equation}
Then, substituting \eqref{equalivance 2} into \eqref{stability dLv 1} gets
\begin{equation} \label{stability dLv 2}
            \dot{L}_{V} 
     = -x^{\top} Q x - \mathcal{W}(u_0+\Delta u^{*}) - \Bar{\xi}^2_o - 2 \beta \tanh^{-1}(\frac{u_0 + \Delta u^{*}}{\beta}) \xi -2 \beta \tanh^{-1}(\frac{u_0 + \Delta u^{*}}{\beta})( \Delta \hat{u}-\Delta u^*).
    \end{equation}
According to \eqref{equalivance 4}-\eqref{epsilon_u_2}, \eqref{stability dLv 2} follows
\begin{equation} \label{stability dLv 3}
            \dot{L}_{V} 
     \leq -x^{\top} Q x - (\Bar{\xi}^2_o-\xi^{\top}\xi ) - [\beta \tanh^{-1}(\frac{u_0 + \Delta u^{*}}{\beta})+\xi]^2+\frac{1}{2} \nabla {V^*}^{\top}\bar{g} \bar{g}^{\top}\nabla V^{*} -2 \beta \tanh^{-1}(\frac{u_0 + \Delta u^{*}}{\beta})( \Delta \hat{u}-\Delta u^*).
\end{equation}
    \emph{The term $-2 \beta \tanh^{-1}(\frac{u_0 + \Delta u^{*}}{\beta})( \Delta \hat{u}-\Delta u^*)$ in \eqref{stability dLv 3} follows}
\begin{equation} \label{stability dLv 4}
-2 \beta \tanh^{-1}(\frac{u_0 + \Delta u^{*}}{\beta})( \Delta \hat{u}-\Delta u^*)
\leq \beta^2 \left\| \tanh^{-1}(\frac{u_0 + \Delta u^{*}}{\beta}) \right\|^2 + \left\| \Delta \hat{u}-\Delta u^* \right\|^2.
\end{equation}
By using \eqref{optimal incremental u}, \eqref{optimal dV approximation}, and the mean-value theorem, the optimal incremental control is rewritten as
\begin{equation} \label{stability dLv 41}
\Delta u^* = -\beta\tanh(\frac{1}{2\beta}\bar{g}^{\top}\nabla \Phi^{\top}W^*)-\epsilon_{\Delta u^*} -u_0,
\end{equation}
where $\epsilon_{\Delta u^*} = \frac{1}{2} (\underline{\mathbf{1}}-\tanh^2(\eta))\bar{g}^{\top}\nabla \epsilon$, and $\eta \in \mathbb{R}^{m}$ is chosen between $\frac{1}{2\beta}\bar{g}^{\top}\nabla \Phi^{\top}W^*$ and $\frac{1}{2\beta}\bar{g}^{\top}\nabla V^*$, $\mathbf{1} = [1,\cdots,1]^{\top}\in \mathbb{R}^m$. According to $\left\| \nabla\epsilon  \right\| \leq b_{\epsilon x}$ in Assumption \ref{bound of NN issues},  $\left\| \epsilon_{\Delta u^*} \right\| \leq \frac{1}{2} \left\|\bar{g}\right\|b_{\epsilon x}$ holds.
Then, by combining \eqref{optimal incremental hat u} with \eqref{stability dLv 41}, we get
\begin{equation} \label{stability dLv 5}
\Delta \hat{u}-\Delta u^* = \beta(\tanh(\frac{1}{2\beta}\bar{g}^{\top}\nabla \Phi^{\top}W^*)-\tanh(\frac{1}{2\beta}\bar{g}^{\top}\nabla \Phi^{\top}\hat{W})+\epsilon_{\Delta u^*}.
\end{equation}
For simplicity, denoting $\mathscr{G}^* = \frac{1}{2\beta}\bar{g}^{\top}\nabla \Phi^{\top}W^*$ and $\hat{\mathscr{G}} = \frac{1}{2\beta}\bar{g}^{\top}\nabla \Phi^{\top}\hat{W}$, where $\hat{\mathscr{G}} = [\hat{\mathscr{G}}_1,\cdots,\hat{\mathscr{G}}_m] \in \mathbb{R}^{m}$ with $\hat{\mathscr{G}}_j \in \mathbb{R}, j=1,\cdots,m$. 
Based on \eqref{optimal incremental u} and \eqref{optimal incremental hat u}, the Taylor series of $\tanh(\mathscr{G}^*)$ follows
\begin{equation} \label{Taylor series}
        \tanh(\mathscr{G}^*) = \tanh(\hat{\mathscr{G}}) +\frac{\partial \tanh(\hat{\mathscr{G}})}{\partial\hat{\mathscr{G}}}(\mathscr{G}^*-\hat{\mathscr{G}})+O((\mathscr{G}^*-\hat{\mathscr{G}})^2)
        =\tanh(\hat{\mathscr{G}})-\frac{1}{2 \beta}(I_{m \times m} - \mathscr{D}(\hat{\mathscr{G}}))  \bar{g}^{\top}
        \nabla \Phi^{\top}\Tilde{W} + O((\mathscr{G}^*-\hat{\mathscr{G}})^2),
\end{equation}
where $\mathscr{D}(\hat{\mathscr{G}}) = diag(\tanh^2(\hat{\mathscr{G}}_1),\cdots,\tanh^2(\hat{\mathscr{G}}_m))$, and $O((\mathscr{G}^*-\hat{\mathscr{G}})^2)$ is a higher order term of the Taylor series. By following \cite[Lemma 1]{yang2016online}, this higher order term is bounded as
    \begin{equation}\label{bound of higher order term}
        \begin{aligned}
            \left\| O((\mathscr{G}^*-\hat{\mathscr{G}})^2)  \right\| & \leq 2\sqrt{m}+\frac{1}{\beta} \left\|\bar{g}\right\| {b_\Phi}_x \left\|\Tilde{W} \right\|.
        \end{aligned}
    \end{equation}
Based on \eqref{Taylor series}, \eqref{stability dLv 5} is rewritten as 
\begin{equation}\label{u-u*-a}
        \Delta \hat{u}-\Delta u^* = \beta(\tanh(\mathscr{G}^*)-\tanh(\hat{\mathscr{G}}))+\epsilon_{\Delta u^*}
        =-\frac{1}{2}(I_{m \times m} - \mathscr{D}(\hat{\mathscr{G}}))  \bar{g}\nabla \Phi^{\top}\Tilde{W}+ \beta O((\mathscr{G}^*-\hat{\mathscr{G}})^2)+\epsilon_{\Delta u^*}.
 \end{equation}
According to \cite{yang2016online}, $\left\|I_{m \times m} - \mathscr{D}(\hat{\mathscr{G}})\right\| \leq 2$ holds. Then, by combining \eqref{bound of higher order term} with \eqref{u-u*-a},  \emph{$\left\| \Delta \hat{u}-\Delta u^* \right\|^2$ in \eqref{stability dLv 4} follows}
    \begin{equation}\label{u-u* abs 2}
        \begin{aligned}
            \left\| \Delta \hat{u}-\Delta u^* \right\|^2 & \leq 3 \beta^2 \left\| O((\mathscr{G}^*-\hat{\mathscr{G}})^2)\right\|^2 +3\left\| \epsilon_{\Delta u^*} \right\|^2
             +3\left\| -\frac{1}{2}(I_{m \times m} - \mathscr{D}(\hat{\mathscr{G}})) \bar{g}^{\top}\nabla \Phi^{\top}\Tilde{W}\right\|^2 \\
             & \leq 6  \left\|\bar{g}\right\|^2 {b^2_\Phi}_x \left\|\Tilde{W}\right\|^2+ 12m\beta^2 + \frac{3}{4} \left\|\bar{g}\right\|^2b^2_{\epsilon x}
             + 12 \beta \sqrt{m}  \left\|\bar{g}\right\| {b_\Phi}_x \left\|\Tilde{W}\right\|.
        \end{aligned}
\end{equation}
Based on \eqref{equalivance 2}, \eqref{optimal dV approximation}, Assumption \ref{bound of NN issues}, and the fact that $\left\|W^{*}\right\| \leq b_{W^{*}}$,
\emph{$\left\| \tanh^{-1}((u_0 + \Delta u^{*}) / \beta) \right\|^2$ in \eqref{stability dLv 4} follows}
\begin{equation} \label{stability dLv 51}
\left\| \tanh^{-1}(\frac{u_0 + \Delta u^{*}}{\beta}) \right\|^2 = \left\| \frac{1}{4\beta^2} \nabla {V^*}^{\top}\bar{g} \bar{g}^{\top}\nabla V^{*} \right\| \leq \frac{1}{4\beta^2} \left\|\bar{g}\right\|^2 {b^2_\Phi}_x b^2_{W^{*}} + \frac{1}{4\beta^2} b^2_{\epsilon x}\left\|\bar{g}\right\|^2+\frac{1}{2\beta^2}\left\|\bar{g}\right\|^2 {b_\Phi}_x b_{\epsilon x} b_{W^{*}}.
\end{equation}
Using \eqref{u-u* abs 2} and \eqref{stability dLv 51}, \eqref{stability dLv 4} reads 
\begin{equation} \label{stability dLv 6}
\begin{aligned}
-2 \beta \tanh^{-1}(\frac{u_0 + \Delta u^{*}}{\beta})( \Delta \hat{u}-\Delta u^*)
 \leq & \frac{1}{4} \left\|\bar{g}\right\|^2 {b^2_\Phi}_x b^2_{W^{*}} + \frac{1}{4} b^2_{\epsilon x}\left\|\bar{g}\right\|^2+\frac{1}{2}\left\|\bar{g}\right\|^2 {b_\Phi}_x b_{\epsilon x} b_{W^{*}}\\
+& 6  \left\|\bar{g}\right\|^2 {b^2_\Phi}_x \left\|\Tilde{W}\right\|^2+ 12m\beta^2 + \frac{3}{4} \left\|\bar{g}\right\|^2b^2_{\epsilon x}
             + 12 \beta \sqrt{m}  \left\|\bar{g}\right\| {b_\Phi}_x \left\|\Tilde{W}\right\|.
 \end{aligned}
\end{equation}
Substituting \eqref{stability dLv 6} into \eqref{stability dLv 3}, finally the first term $\dot{L}_{V}$ follows
\begin{equation} \label{stability dLv 7}
\begin{aligned}
\dot{L}_{V} \leq & -x^{\top} Q x - (\Bar{\xi}^2_o-\xi^{\top}\xi ) - [\beta \tanh^{-1}(\frac{u_0 + \Delta u^{*}}{\beta})+\xi]^2
     + \frac{3}{4} \left\|\bar{g}\right\|^2 {b^2_\Phi}_x b^2_{W^{*}} + \frac{3}{4} b^2_{\epsilon x}\left\|\bar{g}\right\|^2+\frac{3}{2}\left\|\bar{g}\right\|^2 {b_\Phi}_x b_{\epsilon x} b_{W^{*}}\\
     & + 6  \left\|\bar{g}\right\|^2 {b^2_\Phi}_x \left\|\Tilde{W}\right\|^2+ 12m\beta^2 + \frac{3}{4} \left\|\bar{g}\right\|^2b^2_{\epsilon x}
             + 12 \beta \sqrt{m}  \left\|\bar{g}\right\| {b_\Phi}_x \left\|\Tilde{W}\right\|.
\end{aligned}
\end{equation}
\emph{As for the second term $\dot{L}_{W}$}, based on \eqref{w update law} and \eqref{dVcl}, it follows
\begin{equation} \label{stability dLw}
\dot{L}_{W} \leq - \Tilde{W}^{\top} B \Tilde{W} +\Tilde{W}^{\top}\epsilon_{\Tilde{W}}.
\end{equation}
\emph{Finally, as for $\dot{J}$}, substituting \eqref{stability dLv 7} and \eqref{stability dLw} into \eqref{stability dV}, we get 
\begin{equation} \label{stability final result}
\dot{J} \leq  -\mathcal{A}
-\mathcal{B} \left\|\Tilde{W}\right\|^2
+\mathcal{C} \left\|\Tilde{W}\right\| + \mathcal{D},
\end{equation}
where $\mathcal{A} = x^{\top} Q x + (\Bar{\xi}^2_o-\xi^{\top}\xi ) + [\beta \tanh^{-1}(\frac{u_0 + \Delta u^{*}}{\beta})+\xi]^2$,
$\mathcal{B} = \lambda_{\min}(B)- 6\left\|\bar{g}\right\|^2 {b^2_\Phi}_x$,
$\mathcal{C} = 12 \beta \sqrt{m}  \left\|\bar{g}\right\| {b_\Phi}_x +\Bar{\epsilon}_{\Tilde{W}}$,
and $\mathcal{D} = \frac{3}{4} \left\|\bar{g}\right\|^2 {b^2_\Phi}_x b^2_{W^{*}} + \frac{3}{2} b^2_{\epsilon x}\left\|\bar{g}\right\|^2+ \frac{3}{2}\left\|\bar{g}\right\|^2 {b_\Phi}_x b_{\epsilon x} b_{W^{*}}+ 12m\beta^2$.
Let the parameters be chosen such that $\mathcal{B} > 0$. Since $\mathcal{A}$ is positive definite, the above Lyapunov derivative \eqref{stability final result} is negative if
    \begin{equation} \label{negative condition}
        \begin{aligned}
        \left\| \Tilde{W} \right\| > \frac{\mathcal{C}}{2\mathcal{B}}+\sqrt{\frac{\mathcal{C}^2}{4\mathcal{B}^2}+\frac{\mathcal{D} }{\mathcal{B}}}.
        \end{aligned}
    \end{equation}
Thus, the critic weight learning error converges to the residual set
\begin{equation} \label{compact set}
    \Tilde{\Omega}_{\Tilde{W}} = \left\{\Tilde{W} | \left\| \Tilde{W} \right\| \leq \frac{\mathcal{C}}{2\mathcal{B}}+\sqrt{\frac{\mathcal{C}^2}{4\mathcal{B}^2}+\frac{\mathcal{D} }{\mathcal{B}}} \right\}.
\end{equation}
This completes the proof.
\end{proof}

\section{Numerical Simulation}\label{sec simulation}
This section conducts multiple comparative numerical simulations to validate the effectiveness and superiority of our proposed IADP, especially in terms of the reduced control energy expenditure shown in Section \ref{Sim part 1}, and the enhanced robustness illustrated in Section \ref{Sim Part 2}.
Here, we choose the widely investigated pendulum in ADP related works \cite{liu2013policy,si2001online} as a benchmark. The dynamics of the pendulum follows
    \begin{equation}\label{pendulum system}
        \begin{cases}
            \frac{d \theta}{dt} = \vartheta + d \\
            J\frac{d\vartheta}{dt} = u-Mgl\sin{\theta}-f_d\frac{d \theta}{dt},
        \end{cases}
    \end{equation}
    where $\theta, \vartheta \in \mathbb{R}$ denote the angle and the angular velocity of the pendulum, respectively.
    $M = 1/3 \ \mathrm{kg}$ and  $l = 3/2 \ \mathrm{m}$ are the mass and length of the pendulum, respectively. Let $g = 9.8 \left.\mathrm{m}\middle/\mathrm{s}^2\right.$ be the gravity, 
    $J = 4/3 Ml^2 \ \mathrm{kg}\!\cdot\!\mathrm{m}^2$ be the rotary inertia, and $f_d = 0.2$ be the frictional factor. Here $d$ represents an external disturbance.
\subsection{Validation of the reduced control effort of IADP} \label{Sim part 1}
This section compares IADP with the zero-sum game based ADP (ZSADP) \cite{vamvoudakis2012online} and the transformed optimal control based ADP (TADP) \cite{liu2015reinforcement} to verify the superiority of IADP regarding the reduced control effort. Note that among existing ADP related works, model-based ZSADP and TADP are the two most widely adopted methods to deal with the robust stabilization problem illustrated as Problem \ref{Robust stabilization}.
Firstly, to conduct convincing comparative simulations, an often used vanishing (state-dependent) disturbance in ZSADP and TADP related works \cite{vamvoudakis2012online, liu2015reinforcement} is deliberately chosen in Section \ref{sim state dependent disturbance}
to fully show the performance of ZSADP and TADP.
Then, in Section \ref{mutliple disturbance}, except for the vanishing disturbance used in Section \ref{sim state dependent disturbance}, we make a step further by additionally introducing measurement noises, non-vanishing disturbances, and sudden physical changes into the simulation environment. The conducted comparative simulations under multiple sources of uncertainties and disturbances further exemplify the advantage of our proposed IADP in terms of the reduced control effort.

\subsubsection{Validation under the vanishing (state-dependent) disturbance} \label{sim state dependent disturbance}
In this section, by following\cite{wang2014policy}, the chosen state-dependent disturbance follows $d_1 = \omega_{1} \theta \sin(\omega_{2}\vartheta)$, where $\omega_{1}$ and $\omega_{2}$ are randomly chosen within the scope $[-\sqrt{2}/2,\sqrt{2}/2]$ and $[-2,2]$, respectively. 
Let $x_{1} = \theta$ and $x_2 = \vartheta$, the original pendulum system \eqref{pendulum system} is rewritten as
    \begin{equation}\label{pendulum system original system}
        \begin{bmatrix}
            \dot{x}_{1}  \\
            \dot{x}_{2}  \\
        \end{bmatrix} =
        \underbrace{\begin{bmatrix}
            x_{2}  \\
            -4.9\sin{x_1}-0.2x_2  \\
        \end{bmatrix}}_{f(x)}
        +\underbrace{\begin{bmatrix}
            0  \\
            0.25  \\
        \end{bmatrix}}_{g(x)}u 
          +\underbrace{\begin{bmatrix}
            1 \\
            -0.2  \\
        \end{bmatrix}}_{k(x)}\underbrace{\omega_{1} x_1 \sin(\omega_{2}x_2)}_{d_1(x)}.
    \end{equation}
To drive the pendulum \eqref{pendulum system original system} to the equilibrium point even under input saturation ($\beta = 2 $) and the external disturbance $d_1(x)$, the detailed simulation settings for IADP, ZSADP, and TADP are as follows.

For IADP,  we choose $\bar{g} = [0, 0.1]^{\top}$. Its cost function is considered as
\begin{equation}\label{cost fuction of IADP}
    V_I = \int_{t}^{\infty} x^{\top} Q x + \mathcal{W}(u_0+\Delta u) + \Bar{\xi}^2_o \,d\tau,
\end{equation}
where $Q = I_{2\times 2}$, $\mathcal{W}(u_0+\Delta u) = 2 \beta (u_0+\Delta u) \tanh^{-1}((u_0+\Delta u)/\beta)+\beta^2  \log(1-(u_0+\Delta u)^2/\beta^2)$, and $\Bar{\xi}_o = 2 \left\|\Delta u \right\|$.
The approximate optimal incremental control $\Delta \hat{u}$ and the approximate optimal control $\hat{u}$ follow \eqref{optimal incremental hat u} and \eqref{optimal  hat u}, respectively. IADP requires no explicit model or environmental information except for a predefined constant matrix $\bar{g}$. 

For ZSADP, by following \cite{vamvoudakis2012online}, its cost function follows
\begin{equation}\label{cost fuction of ZSADP}
    V_{Z} = \int_{t}^{\infty} x^{\top} Q x + \mathcal{W}(u_{Z}) - \gamma  d^{\top}_{Z} d_{Z} \,d\tau, 
\end{equation}
where $\mathcal{W}(u_{Z}) = 2\beta u_{Z} \tanh^{-1}(u_{Z}/\beta)+\beta^2 \log(1-u_{Z}^2/\beta^2)$, $\gamma =1$.
For this case, the approximate optimal control policy follows $\hat{u}_{Z} = - \beta \tanh(\frac{1}{2\beta}g^{\top}\nabla \Phi^{\top}\hat{W}_{Z})$,
and the approximate worst-case disturbance policy is $\hat{d}_{Z} = \frac{1}{2\gamma^2} k^{\top}\nabla \Phi^{\top}\hat{W}_{Z}$. Here $\hat{u}_{Z}$ and $\hat{d}_{Z}$ depend on the concert $g(x)$ and $k(x)$ in \eqref{pendulum system original system}, respectively.

For TADP, according to \cite{liu2015reinforcement}, the corresponding cost function follows
\begin{equation}\label{cost fuction of TADP}
    V_{T} = \int_{t}^{\infty} x^{\top} Q x + \mathcal{W}(u_T) + \rho v^{\top}_T v_T+ l^{2}_{M}+d^{2}_{M} \,d\tau,
\end{equation}
where $\rho = 0.1$. The chosen disturbance satisfies $\left\| d(x)\right\| \leq \sqrt{2}/2 \left\| x\right\|$. Thus, $d_{M} = \sqrt{2}/2 \left\| x\right\|$ and $l_{M} = 0.4\sqrt{2} \left\| x\right\|$ are chosen to address the disturbance $d_1(x)$. Much more details can be found in \cite{liu2015reinforcement}. The approximate optimal control follows $\hat{u}_{T} = - \beta \tanh(\frac{1}{2\beta}g^{\top}\nabla \Phi^{\top}\hat{W}_T)$, and the approximate pseudo control follows $\hat{v}_{T} = - \frac{1}{2\rho} h^{\top}\nabla \Phi^{\top}\hat{W}_T$, where $h = (I_{2 \times 2}-gg^{+})k$. For TADP, the explicit knowledge of $g(x)$ and $k(x)$ in \eqref{pendulum system original system} is required to construct $\hat{u}_{T}$ and $\hat{v}_{T}$. 

The aforementioned IADP, ZSADP, and TADP all adopt the single critic structure and our developed off-policy weight update law \eqref{w update law}. To achieve a fair comparison, simulation parameters for three methods are set as same, which is detailly clarified as follows.
To get the approximate solutions to the above value functions \eqref{cost fuction of IADP}-\eqref{cost fuction of TADP}, $\Phi(x) = [x^{2}_{1},x_{1}x_{2},x^{2}_{2},x^{3}_{2},x_{1}x^{2}_{2},x^{2}_{1}x_{2}]^{\top}$ is chosen.
To guarantee the weight convergence, parameters are set as $P = 8$, $\Gamma = 10^{-4}I_{6 \times 6}$, $k_c = 5$, and $k_e= 3$.
The initial values are chosen as $x(0) = [2,-2]^{\top}$, $\hat{u}(0) = 0$,  $\hat{d}_{Z}(0) = 0$ (for ZSADP), and $\hat{v}_{T}(0) = 0$ (for TADP). 
Note that to achieve a fair comparison, we also fix the values of $\omega_{1}$ and $\omega_{2}$ in $d_1(x)$ as a set of randomly selected values: $\omega_{1} = -0.3906$, $\omega_{2} = 1.0051$.
 
The critic ANN weigh convergence result for IADP, ZSADP, and TADP are displayed in 
Fig.\ref{weight three cases}.
Based on the developed off-policy weight update law \eqref{w update law}, the weight convergence is guaranteed without adding external noises to control inputs to achieve the required sufficient exploration. 
\begin{figure}[t]
\centering
{\includegraphics[width=150pt,height=10pc]{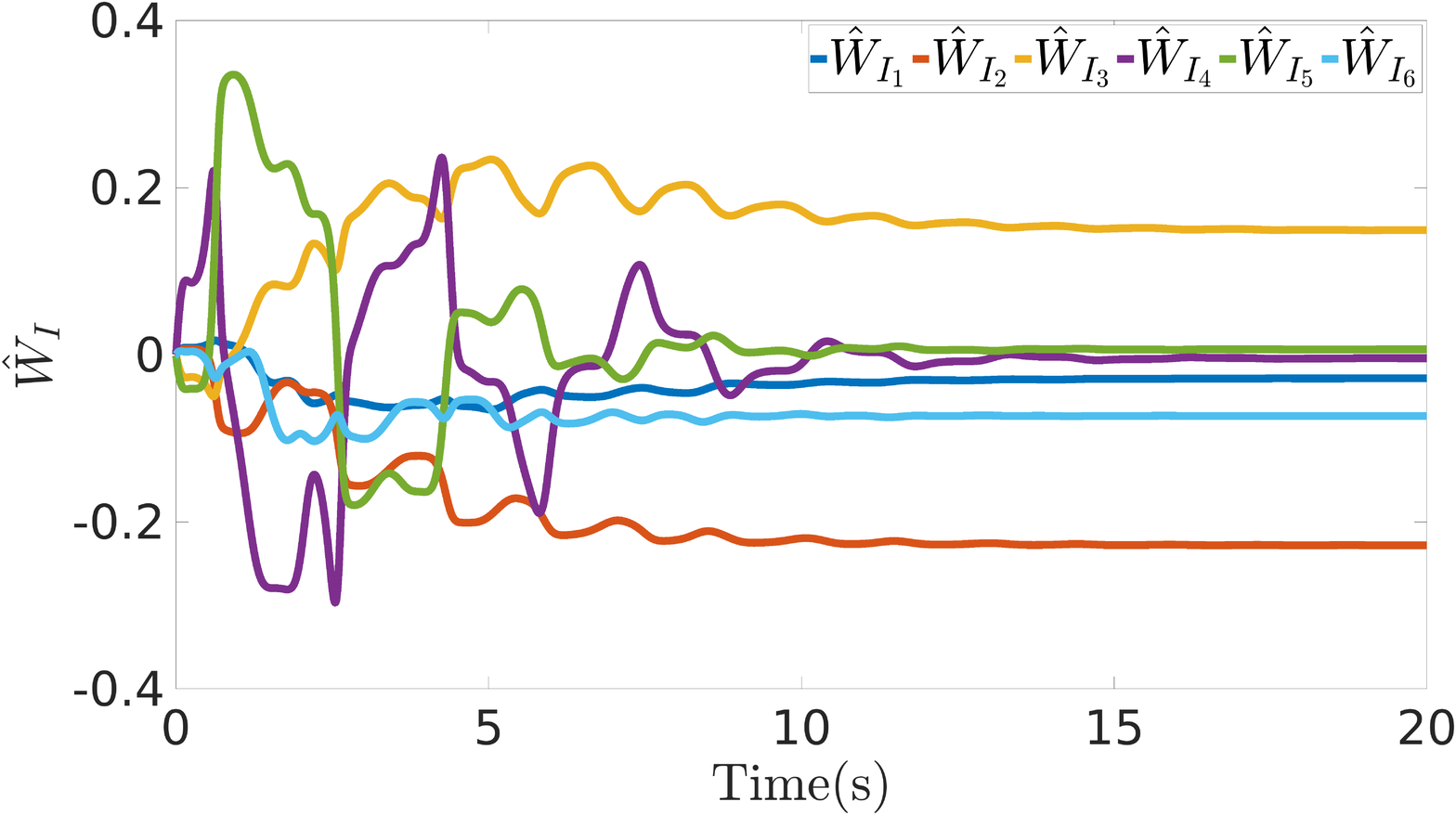}}
{\includegraphics[width=150pt,height=10pc]{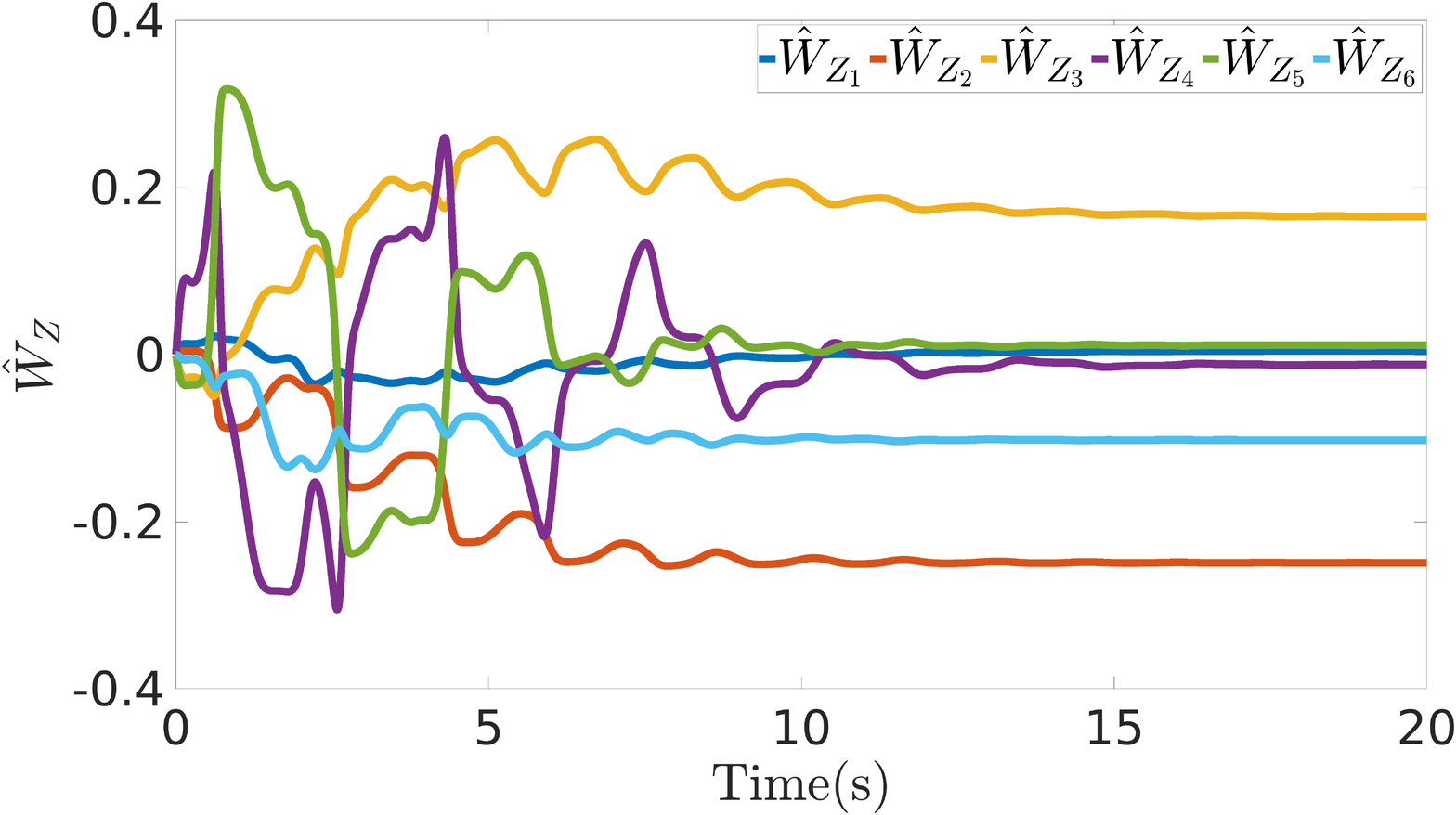}}
{\includegraphics[width=150pt,height=10pc]{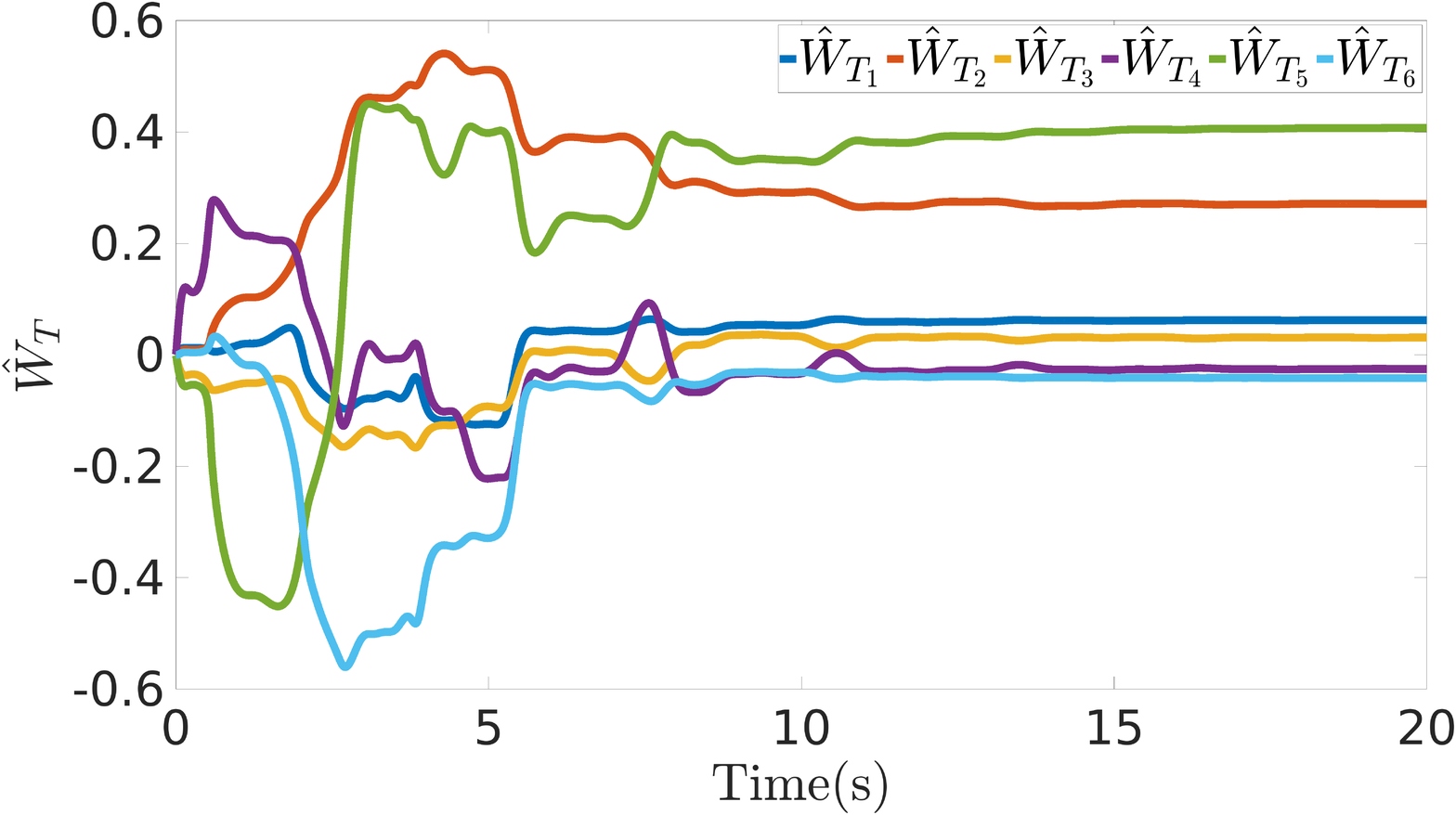}}
\caption{The estimated weight trajectories of IADP ($\hat{W}_I$) , ZSADP ($\hat{W}_Z$), and TADP ($\hat{W}_T$) under the disturbance $d_1(x)$.\label{weight three cases}}
\end{figure}
The state and control trajectories of three cases are shown in Fig.\ref{x and u trajectory}, where the pendulum is successfully driven to the equilibrium point without violating input constraints. However, regarding the peak points of state and control trajectories, the fluctuation range of IADP is smaller than ZSADP and TADP.
\begin{figure}[t]
    \centering
    {\includegraphics[width=150pt,height=10pc]{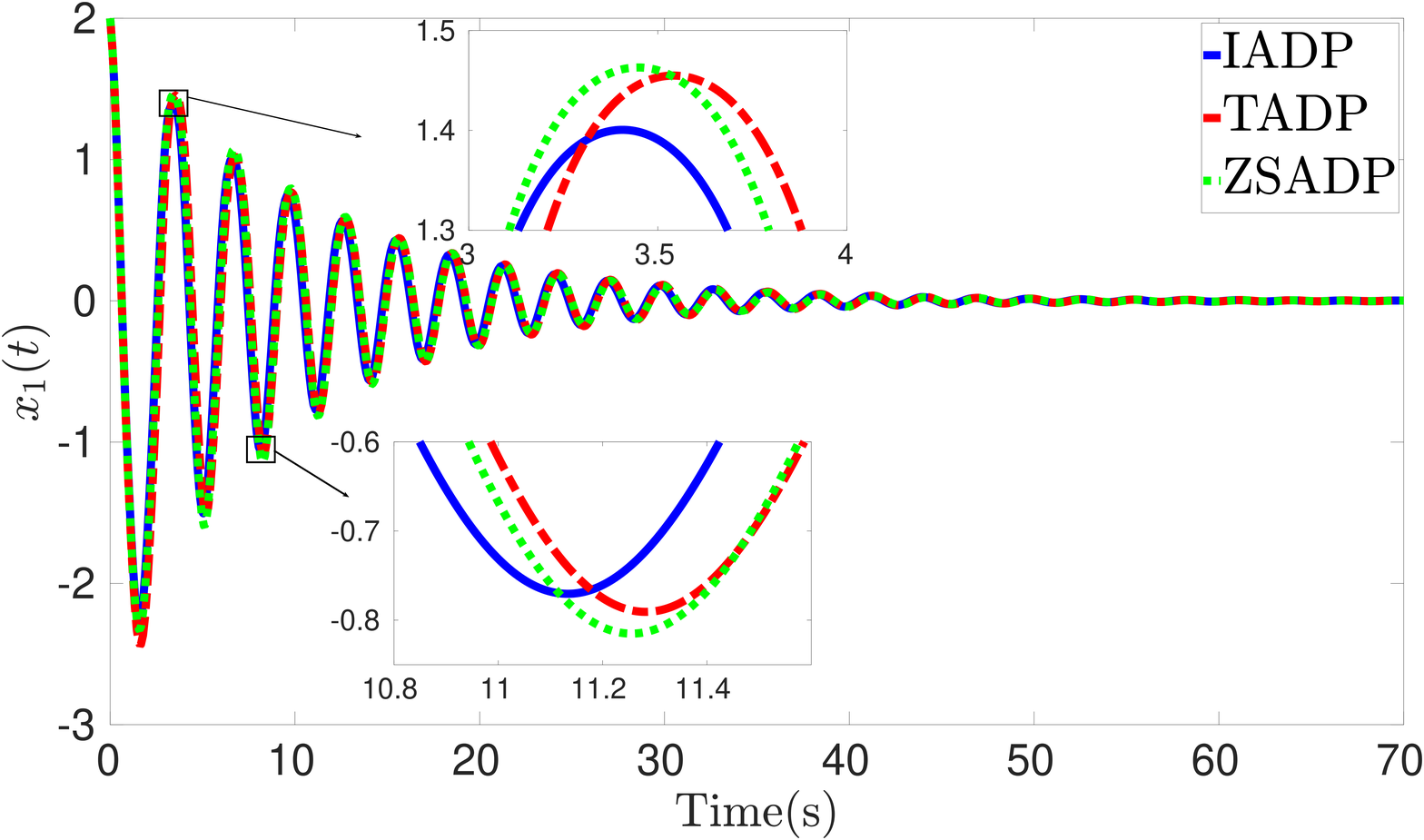}}
   {\includegraphics[width=150pt,height=10pc]{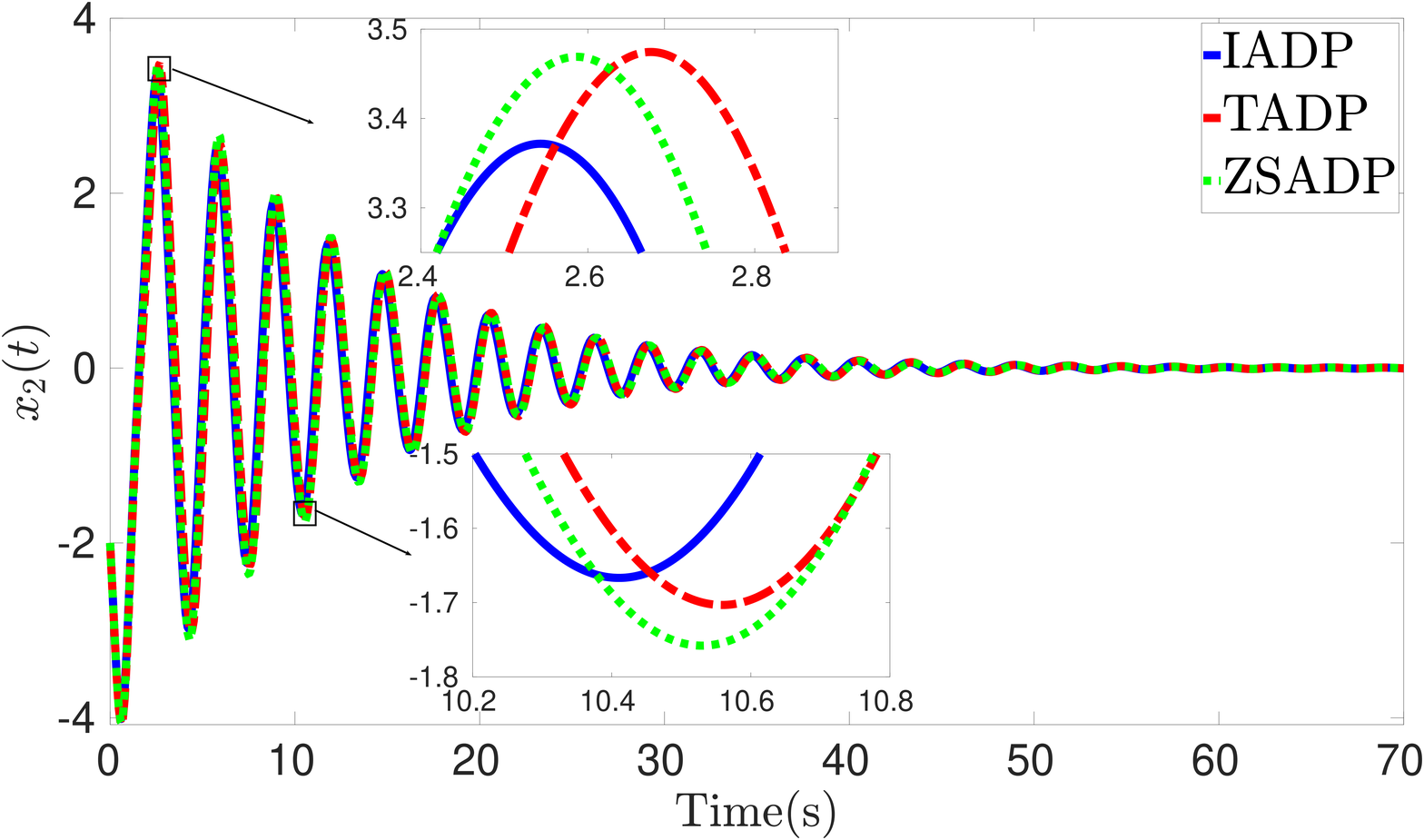}}
   {\includegraphics[width=150pt,height=10pc]{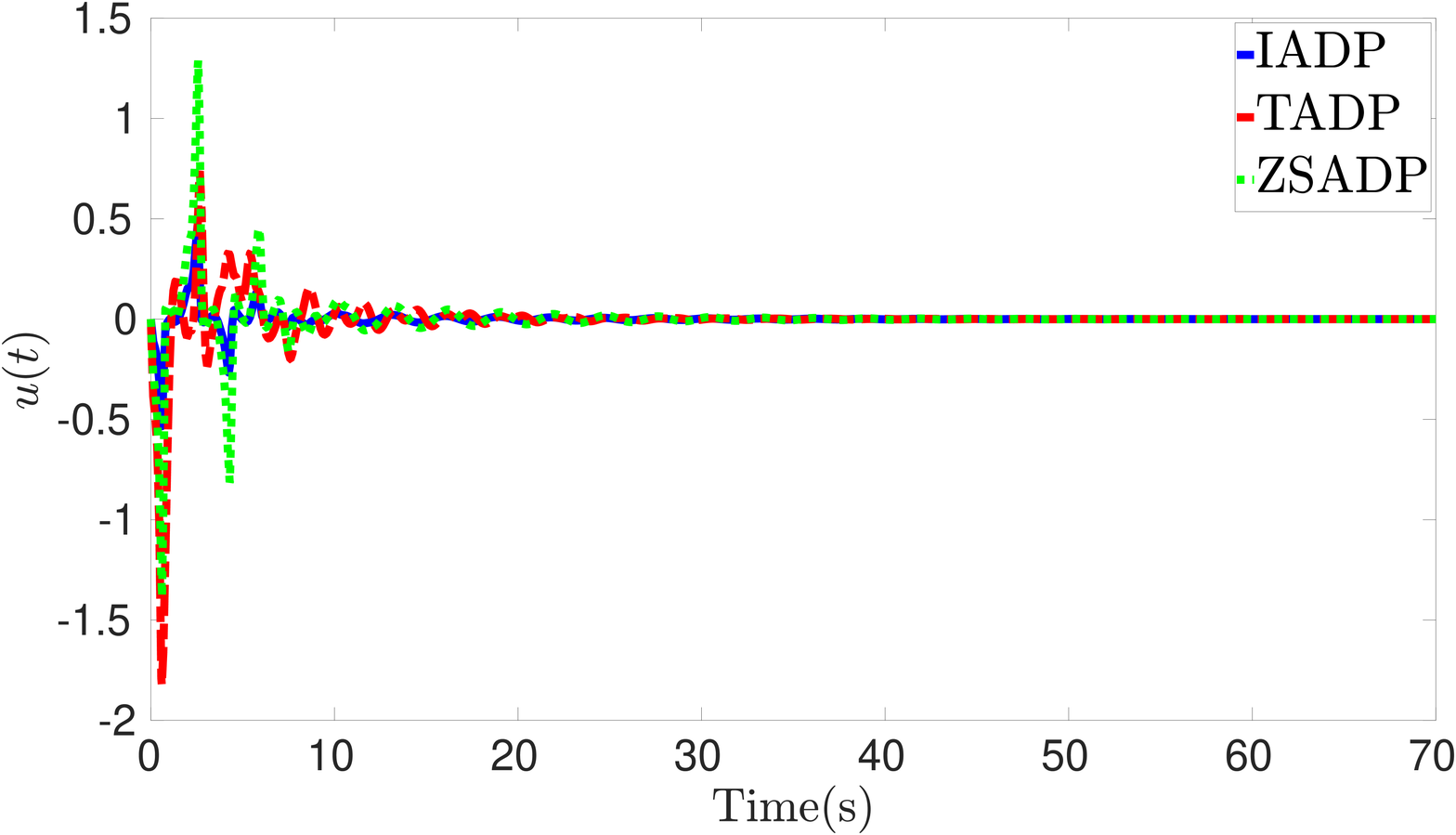}}
    \caption{The state and control trajectories of IADP, ZSADP, and TADP under the vanishing disturbance $d_1(x)$. \label{x and u trajectory}}
\end{figure}

To reveal the superiority of IADP over ZSADP and TADP, we display their corresponding control energy expenditure $E_u = \int_{0}^{\infty}  \left\|\hat{u}\right\|^2  \,d\tau$, and state deviation $E_x = \int_{0}^{\infty}  \left\|x\right\|^2  \,d\tau$ in Fig.\ref{performance trajectory}. 
It is observed in Fig.\ref{performance trajectory} that IADP enjoys a noticeable reduction in utilized control effort, i.e., energy efficiency is highly improved. This makes IADP a more suitable choice for energy-limited platforms. 
This significant decrease in the control effort comes from the achievable optimization of the incremental control inputs. 
Specifically, given $\hat{u}(0)=0$ and IADP prefers a small $\Delta \hat{u}$ at each optimization step, a small $\hat{u}$ is generated to stabilize the pendulum. Thus, we finally get a small $E_u$, which is a cumulative sum of squared $\hat{u}$.
The performance analysis shown in Fig.\ref{performance trajectory} also clarifies the conservativeness of ZSADP and TADP. Although the worst-case disturbance related terms (i.e., $d^{\top}_{Z} d_{Z}$ for ZSADP, $v^{\top}_T v_T$, $l^{2}_{M}$, and $d^{2}_{M}$ for TADP) incorporated into the cost functions \eqref{cost fuction of ZSADP}-\eqref{cost fuction of TADP} allow controller designers to address the additive disturbance $d_1(x)$,
these additionally introduced terms trade off the desired performance indexes of control efforts and state deviations. Thus, a performance compromise problem arises.
\begin{figure}[t]
    \centering
    {\includegraphics[width=200pt,height=12pc]{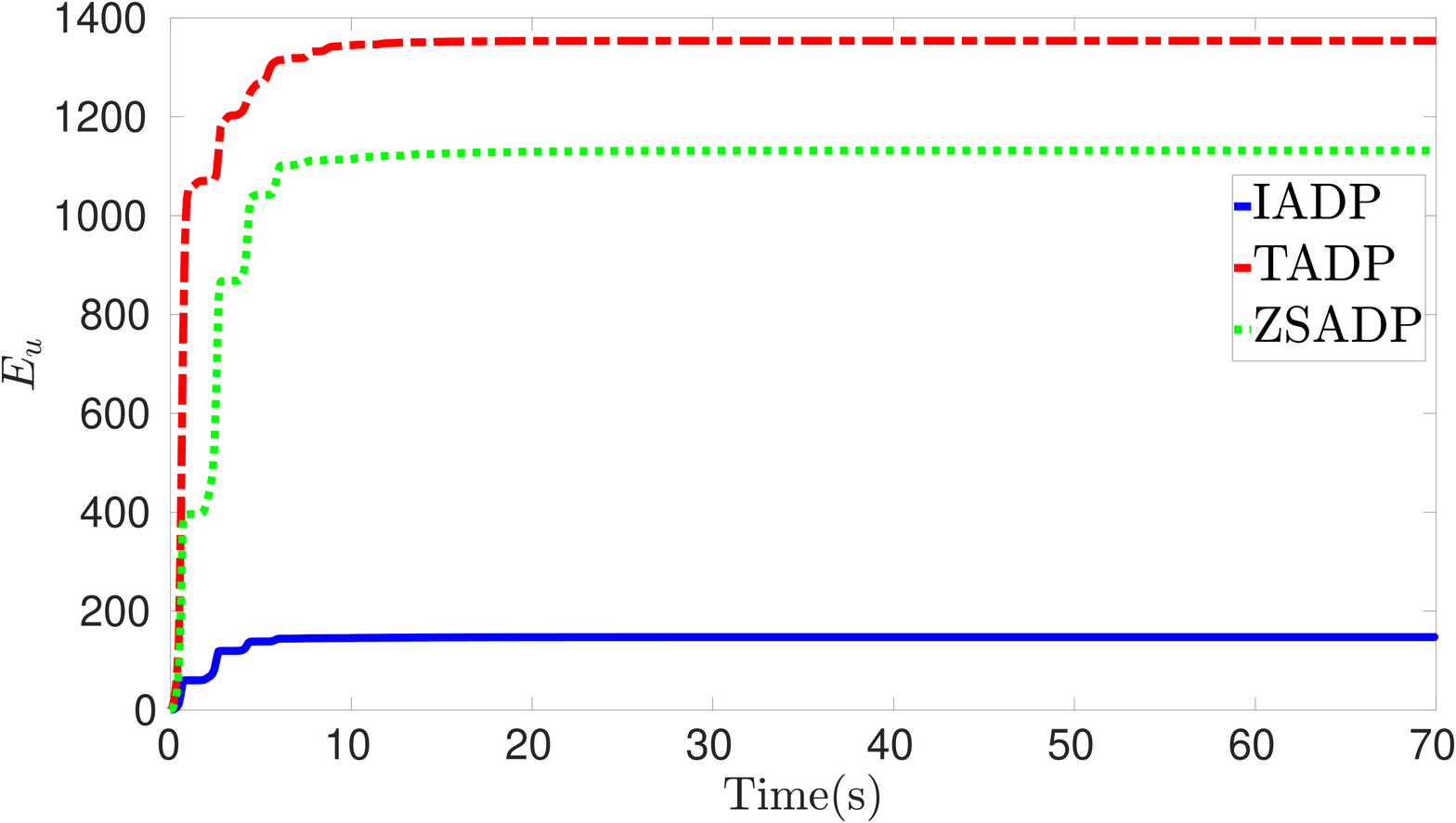}}
   {\includegraphics[width=200pt,height=12pc]{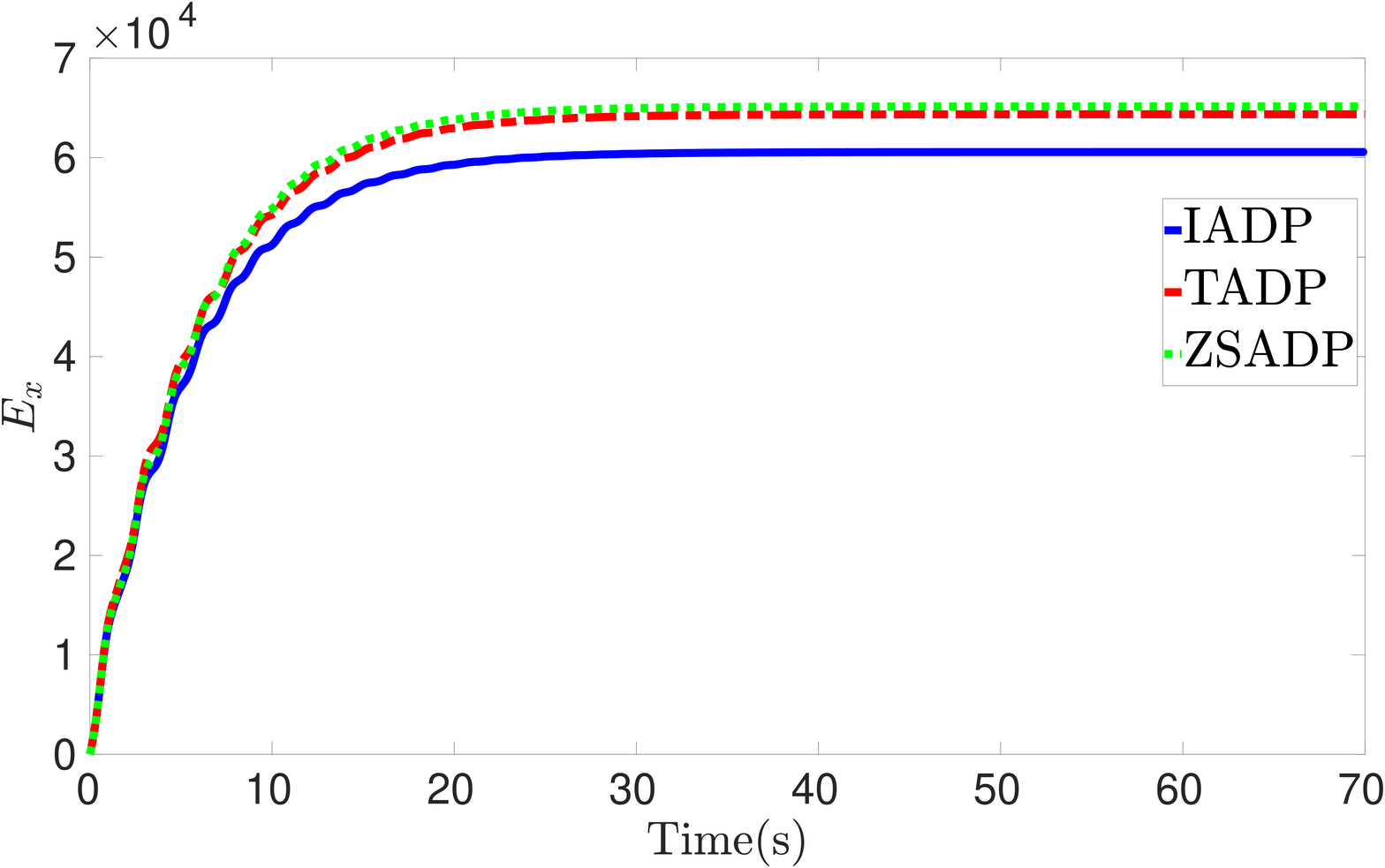}}
    \caption{The performance comparison between IADP, ZSADP, and TADP under the vanishing disturbance $d_1(x)$. \label{performance trajectory}}
\end{figure}
Given the aforementioned simulation results, we know that, even though model-based ZSADP and TADP could provide the rigorous robustness guarantee under worst-case disturbances, they perform poorly than our proposed model-free IADP, especially regarding the control energy expenditure.
\subsubsection{Validation under the non-vanishing disturbance, measurement noise, and physical change}\label{mutliple disturbance}
To further validate the superiority of our proposed IADP over ZSADP and TADP, 
this section conducts comparative simulations under the vanishing disturbance used in Section \ref{sim state dependent disturbance}, as well as the newly introduced non-vanishing disturbance, measurement noise, and sudden physical change.
It is worth noting that ZSADP \cite{vamvoudakis2012online} and TADP \cite{liu2015reinforcement} can only deal with the state-dependent disturbance in a closed-loop form. Thus, here the amplitudes of the chosen non-vanishing disturbance, measurement noise, and physical change are purposely set to be the level that could be tackled by the inherent robustness of ZSADP and TADP. 

Here we follow the time-varying non-vanishing disturbance from \cite{jankovic2018robust}, which is denoted as $d_2(t)$ and set as a square wave with amplitude $0.2$ and period $5 s$.
The measurement noise is chosen as a white Gaussian noise with $50$ dBW.
We add $d_2(t)$ and the measurement noise into the simulation environment during the time from $20 s$ to $60 s$.
To simulate a sudden physical change, e.g., parameter perturbations due to unknown loads put on the pendulum, at $t = 20 s$, the dynamics of pendulum \eqref{pendulum system original system} is randomly reset as
 \begin{equation}\label{pendulum physical changes 1}
        \begin{bmatrix}
            \dot{x}_{1}  \\
            \dot{x}_{2}  \\
        \end{bmatrix} =
        \underbrace{\begin{bmatrix}
            x_{2}  \\
            -2\sin{x_1}-0.1x_2  \\
        \end{bmatrix}}_{f(x)}
        +\underbrace{\begin{bmatrix}
            0  \\
            0.1 \\
        \end{bmatrix}}_{g(x)}u 
          +\underbrace{\begin{bmatrix}
            1 \\
            -0.1  \\
        \end{bmatrix}}_{k(x)}d_1(x).
    \end{equation}
The parameter settings for IADP, ZSADP, and TADP are the same as the settings in Section \ref{sim state dependent disturbance}.

Under the vanishing disturbance $d_1(x)$, the non-vanishing square wave disturbance $d_2(t)$, the white Gaussian measurement noise, and the sudden physical change from \eqref{pendulum system original system} to \eqref{pendulum physical changes 1}, the simulation results are shown in Fig.\ref{x trajectory S2} and Fig.\ref{performance trajectory S2}. The state trajectories displayed in Fig.\ref{x trajectory S2} reveal that three methods all successfully stabilize the pendulum without retuning parameters, i.e., these three methods possess inherent robustness to the aforementioned uncertainties and disturbances in certain amplitudes.
\begin{figure}[t]
    \centering
    {\includegraphics[width=200pt,height=12pc]{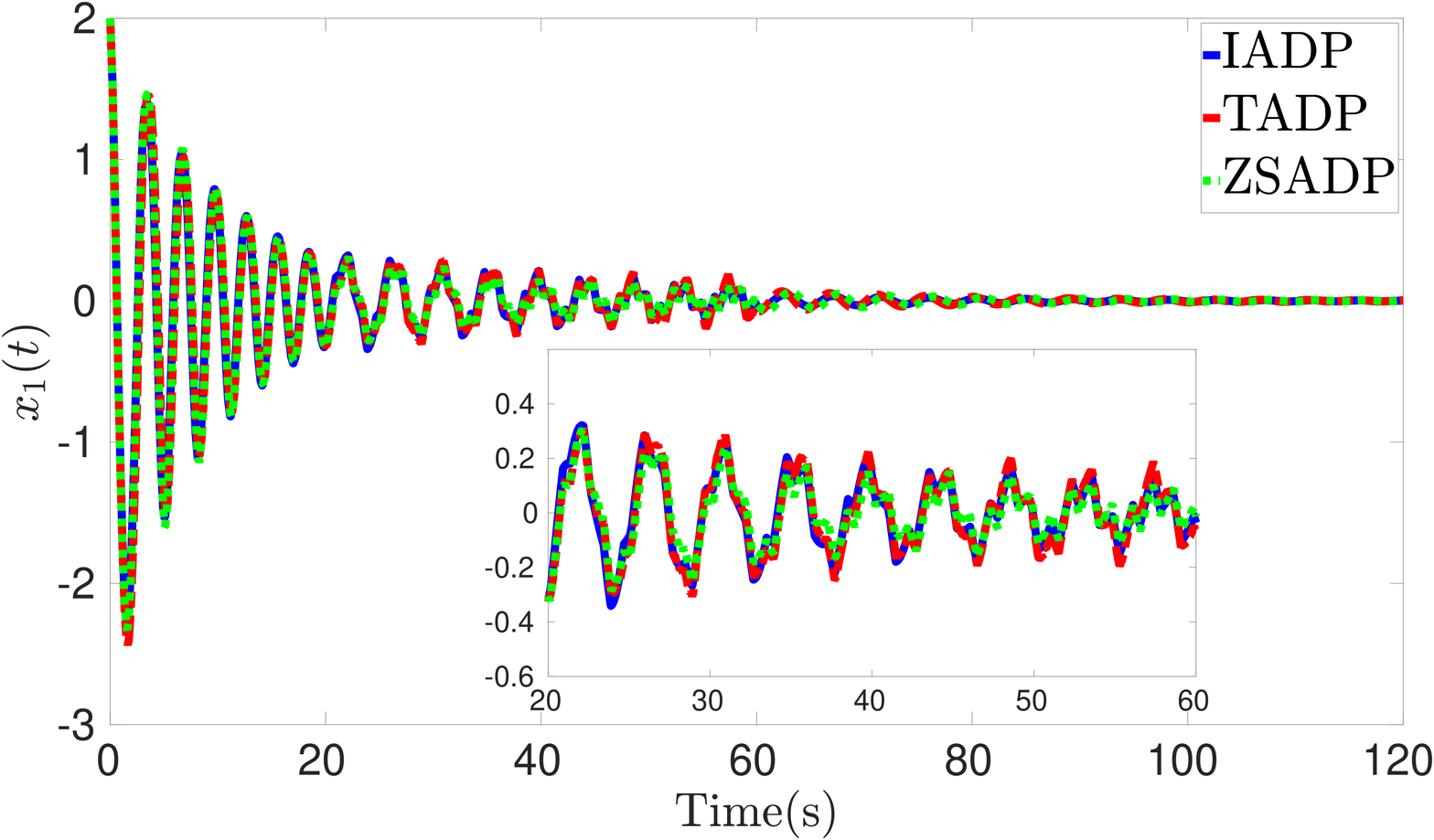}}
   {\includegraphics[width=200pt,height=12pc]{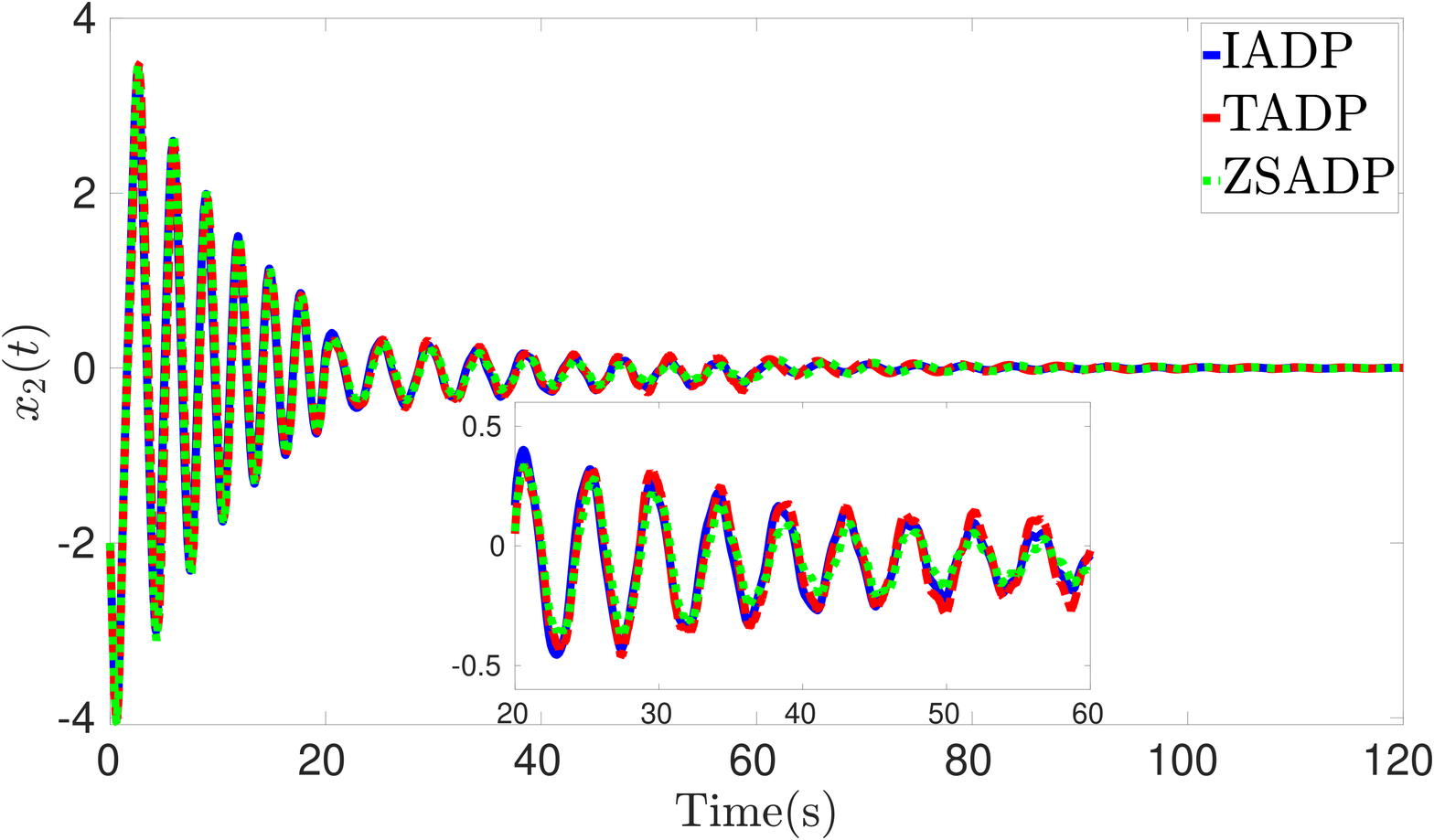}}
    \caption{The state trajectories of IADP, ZSADP, and TADP under the non-vanishing disturbance $d_2(t)$. \label{x trajectory S2}}
\end{figure}
The performance comparison shown in Fig.\ref{performance trajectory S2} further validates the significant control energy deduction of our developed IADP. Comparing to Fig.\ref{performance trajectory} in Section \ref{sim state dependent disturbance}, the increased control effort of IADP results from the required additional control energy to deal with the newly introduced uncertainties and disturbances.
Besides, Fig.\ref{performance trajectory S2} also displays that IADP outperforms ZSADP and TADP in terms of the state deviation $E_x$.
\begin{figure}[t]
    \centering
    {\includegraphics[width=200pt,height=12pc]{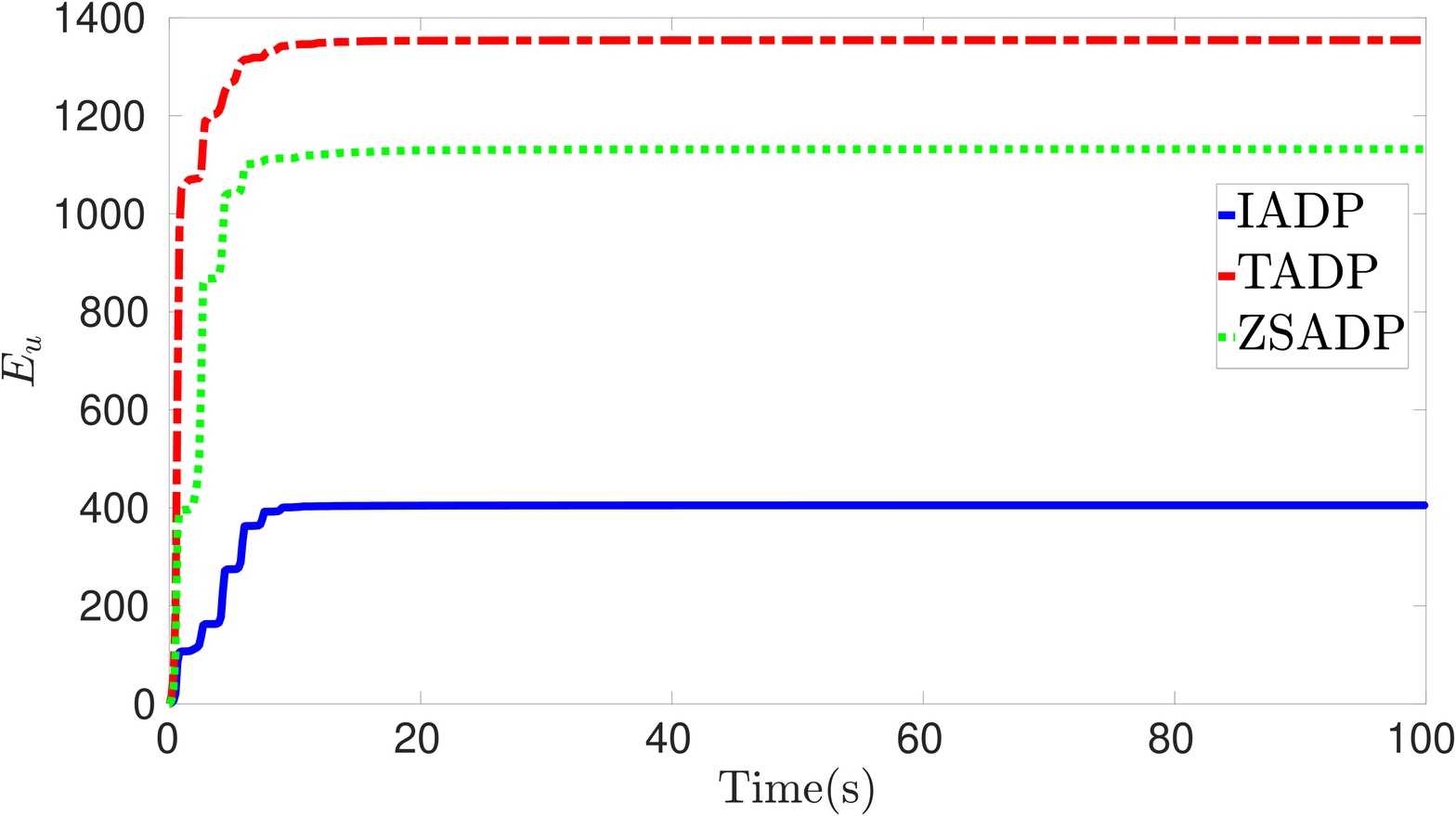}}
   {\includegraphics[width=200pt,height=12pc]{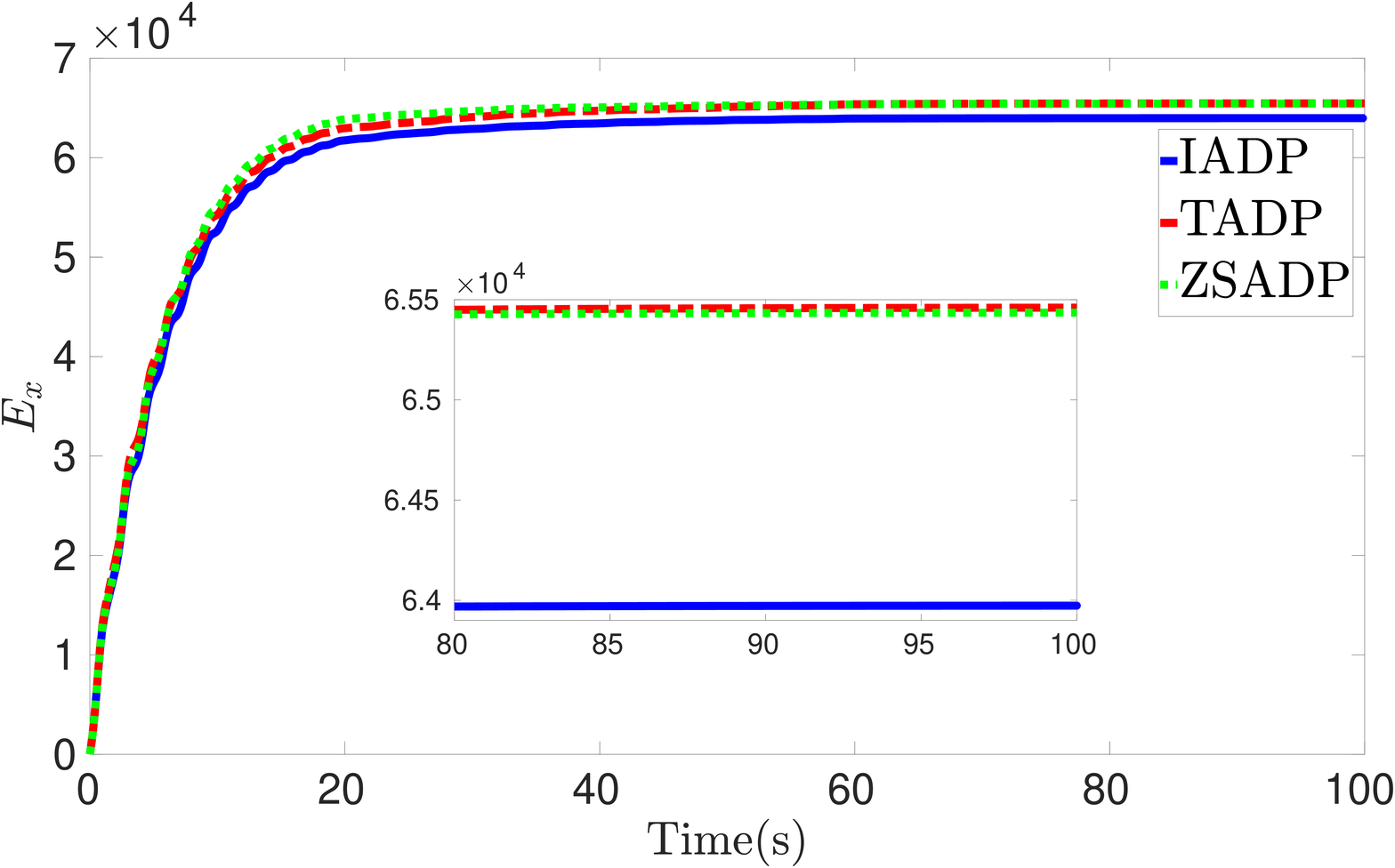}}
    \caption{The performance comparison between IADP, ZSADP, and TADP under the non-vanishing disturbance $d_2(t)$. \label{performance trajectory S2}}
\end{figure}
\subsection{Validation of the enhanced robustness of IADP} \label{Sim Part 2}
To highlight the enhanced robustness of our proposed IADP, this section conducts numerical simulations under a more complex simulation environment comparing to Section \ref{mutliple disturbance}. The details are as follows:
during the time from $20 s$ to $60 s$, the added non-vanishing disturbance $d_3(t)$ is a square wave with amplitude $0.5$ and period $1 s$, whose amplitude and frequency are both improved comparing to $d_2(t)$ used in Section \ref{mutliple disturbance}; the incorporated measurement noise is set as a white Gaussian noise with $10$ dBW, whose magnitude is $5$ times larger than the one chosen in Section \ref{mutliple disturbance}.
Besides, to model a significant physical change, at $t = 20 s$, the pendulum \eqref{pendulum system original system} is reset as
 \begin{equation}\label{pendulum physical changes}
        \begin{bmatrix}
            \dot{x}_{1}  \\
            \dot{x}_{2}  \\
        \end{bmatrix} =
        \underbrace{\begin{bmatrix}
            -x_{2}  \\
            4.9\sin{x_1}-0.2x_2  \\
        \end{bmatrix}}_{f(x)}
        +\underbrace{\begin{bmatrix}
            0  \\
            -0.25 \\
        \end{bmatrix}}_{g(x)}u 
          +\underbrace{\begin{bmatrix}
            1 \\
            -0.2  \\
        \end{bmatrix}}_{k(x)}d_1(x).
    \end{equation}
Comparing to Section \ref{mutliple disturbance}, the simulated physical change here is more aggressive by inverting the sign of model parameters.
The parameter settings for IADP, ZSADP, and TADP follow the settings in Section \ref{sim state dependent disturbance}.

The estimated weight trajectory of IADP shown in Fig.\ref{weight u trajectory S3} illustrates that under multiple sources of uncertainties and disturbances, our proposed off-policy weight update law \eqref{w update law} enables us to collect the real-time data in time and finally achieve the weigh converge.
The control trajectories shown in Fig.\ref{weight u trajectory S3}, and the state trajectories displayed in Fig.\ref{x trajectory S3} clarify the enhanced robustness of IADP. Specifically, IADP  successfully stabilizes the pendulum system under multiple sources of uncertainties and disturbances, however, the robustness of ZSADP and TADP are not enough to tackle such a complex environment. Thus, the control inputs and states of ZSADP and TADP diverge far away immediately after the simulation environment significantly changes at $t = 20s$. 
\begin{figure}[t]
    \centering
    {\includegraphics[width=200pt,height=12pc]{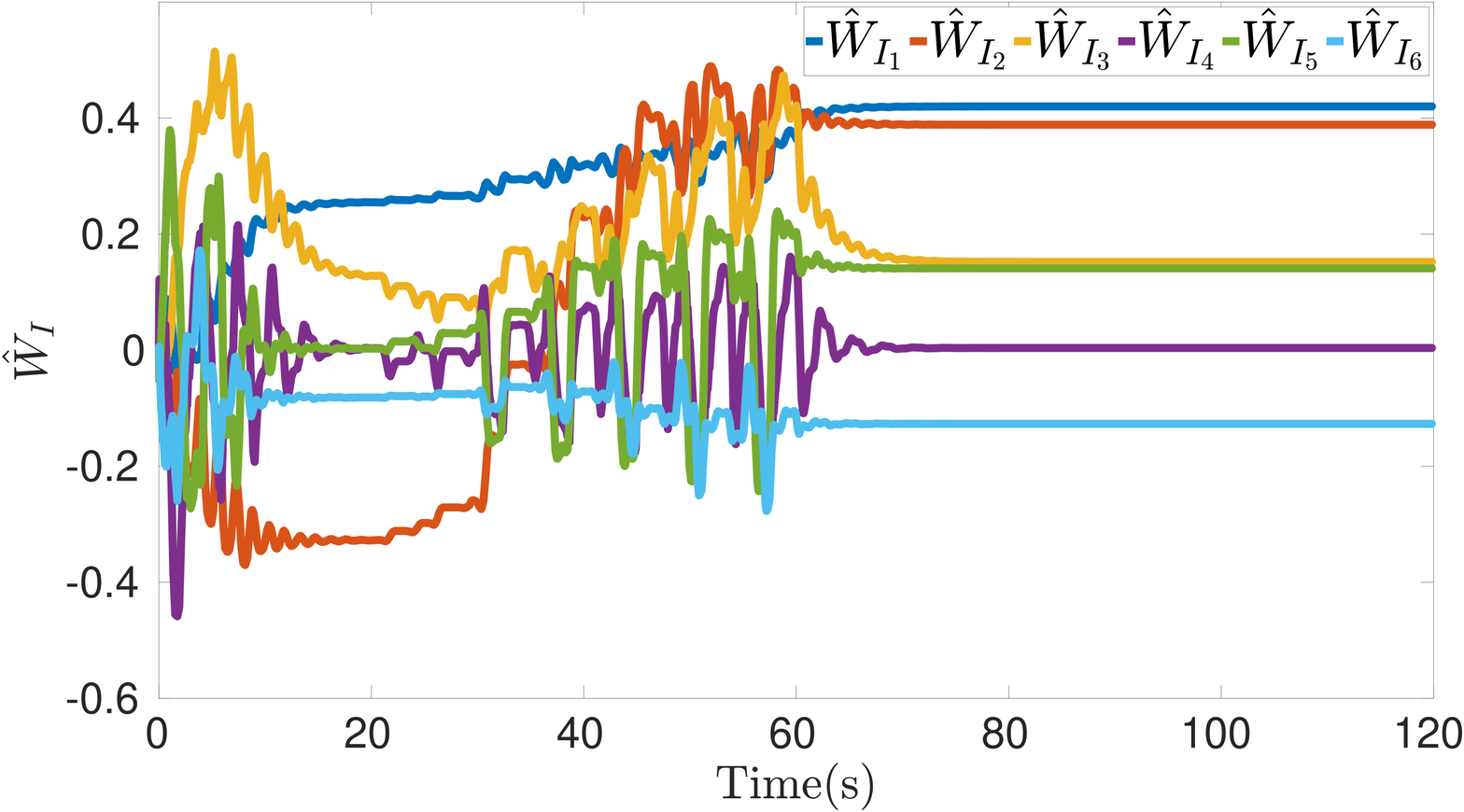}}
   {\includegraphics[width=200pt,height=12pc]{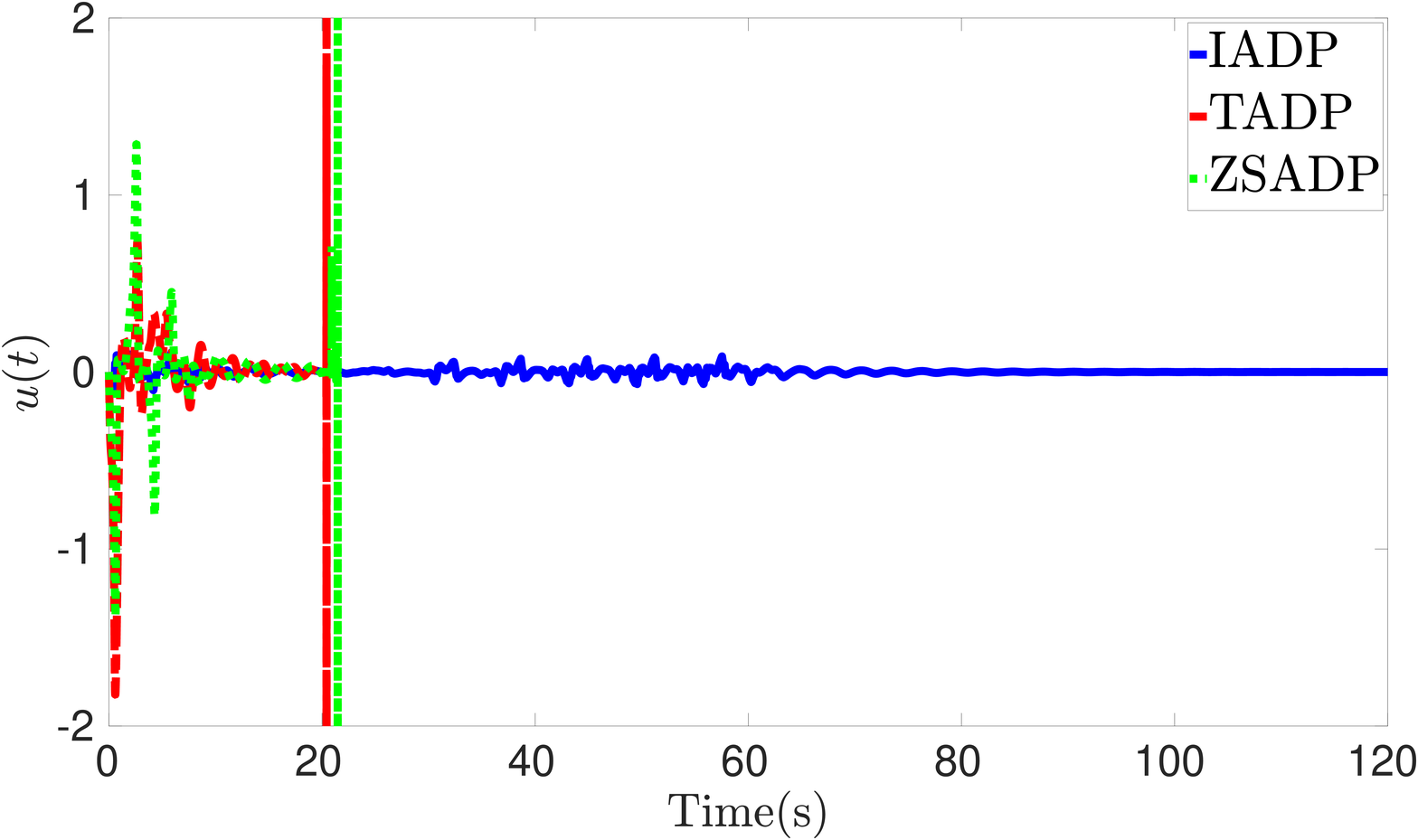}}
    \caption{The estimated weight trajectory of IADP and the control trajectories of IADP, ZSADP, and TADP.\label{weight u trajectory S3}}
\end{figure}

\begin{figure}[t]
    \centering
    {\includegraphics[width=200pt,height=12pc]{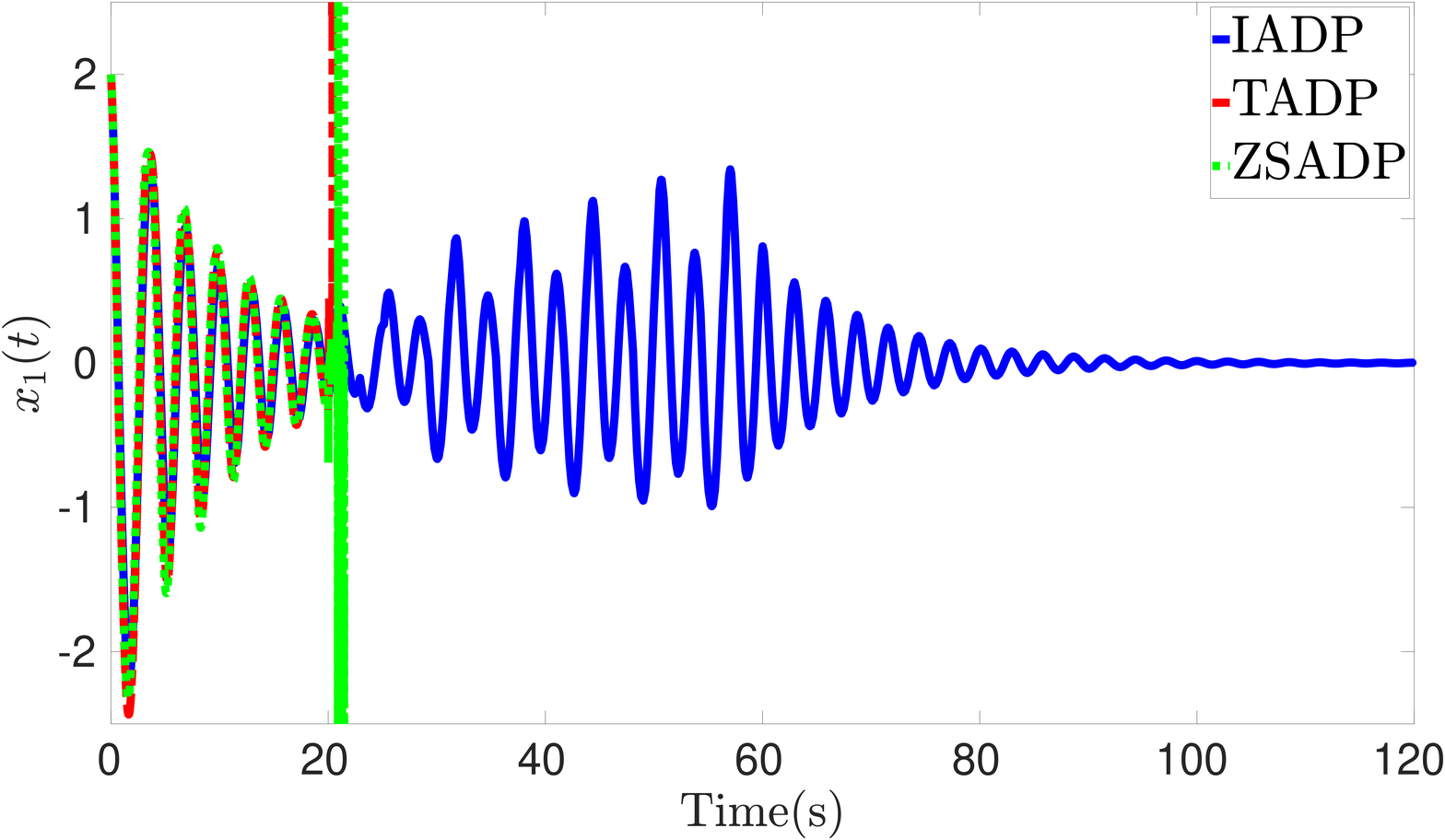}}
   {\includegraphics[width=200pt,height=12pc]{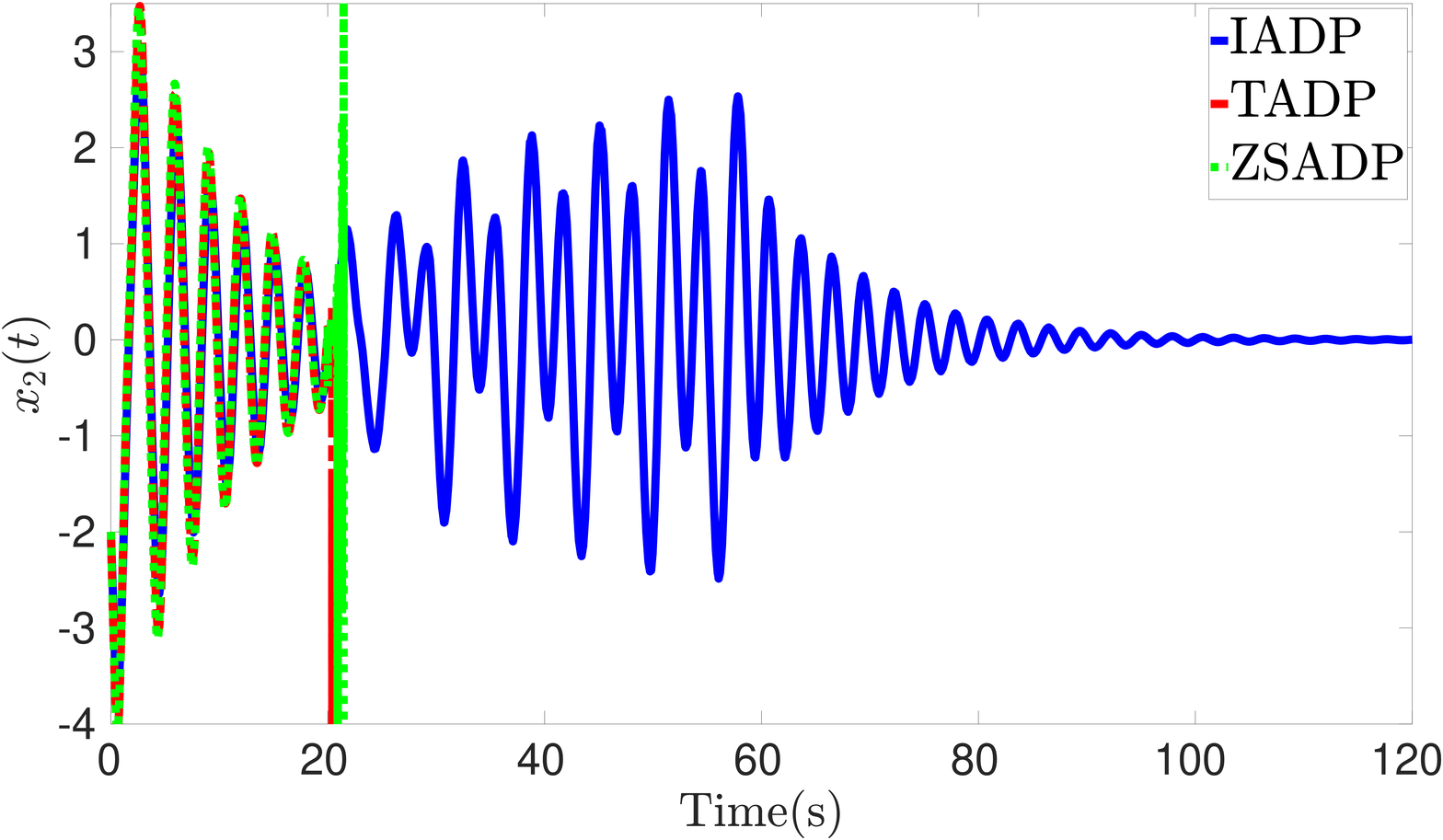}}
    \caption{The state trajectories of IADP, ZSADP, and TADP under a complex simulation environment.\label{x trajectory S3}}
\end{figure}

\section{Conclusion}\label{sec conclusion}
The paper designs an efficient and low-cost model-free control strategy for robust optimal stabilization of continuous-time nonlinear systems. 
To reduce dependence on accurate mathematical models, the TDE technique permits us to get an measured data based incremental dynamics, which is an equivalent of the original dynamics, without requiring any explicit model knowledge or any computation-intensive identification procedures.
Then, the HJB equation, which is constructed based on the incremental dynamics, is approximately solved through a single critic structure. The resulting approximate optimal incremental control strategy stabilizes the controlled system incrementally.
Besides, by transforming the critic ANN weight learning as a parameter identification process and further exploiting the collected experience data, we develop an efficient weight update law with guaranteed weight convergence.
The simultaneous consideration of stability, optimality, and robustness, the introduced simplified single critic structure, and the easily implemented off-policy weight update law make our proposed IADP be promising for practical applications.
Multiple conducted numerical simulations have shown that our proposed IADP outperforms common ADP methods in terms of the reduced control efforts and the enhanced robustness.
Future research works will extend our developed IADP to the robot manipulator optimal tracking control problem with experimental implementations.
In addition, since the efficacy of IADP depends on accurate sensor measurements, we will investigate and address the influence of sensor biases or delays to IADP.

\bibliographystyle{unsrt}  
\bibliography{references}

\end{document}